\documentclass[12pt]{article}
\usepackage{amsfonts}
\usepackage{amsmath, amsfonts, amsthm, amssymb, amscd, mathrsfs}

 \usepackage[cal=dutchcal]{mathalfa}

\usepackage{color}
\theoremstyle{plain}
\newtheorem{theorem}{Theorem}[section]
\newtheorem{proposition}[theorem]{Proposition}

\newtheorem{lemma}[theorem]{Lemma}

\newtheorem{definition}[theorem]{Definition}
\newtheorem{remark}[theorem]{Remark}
\numberwithin{equation}{section}
\let\oldmarginpar\marginpar
\renewcommand\marginpar[1]{\- \oldmarginpar[\raggedleft\footnotesize #1]%
{\raggedright\footnotesize #1}}
\newtheoremstyle{mainthm}
  {3pt}{3pt}% space above/below
  {\itshape}% body font
  {}% indent
  {\bfseries}% head font
  {:}% head punctuation (colon)
  {1em}% after head
  {}% head spec
\theoremstyle{mainthm}
\newtheorem*{maintheorem}{Main Theorem} % unnumbered
\theoremstyle{plain} % restore your default style afterwards

\usepackage{a4wide}

\usepackage[normalem]{ulem}

\newcommand \bse {\begin{subequations}}
\newcommand \ese {\end{subequations}}

\usepackage{lineno, blindtext}

\newcommand \bei {\begin{itemize}}
\newcommand \eei {\end{itemize}}

\newcommand \Rstar {R_\star}
\newcommand \Vstar {V_\star}

\newcommand \Wstar {W_\star}
\newcommand \be         {\begin{equation}}
\newcommand \bel {\be\label}

\newcommand \eps \epsilon

\newcommand \coeff \kappa

\newcommand \Mcal {\mathcal M}

\newcommand \hb {\overline h}
\newcommand \lb {\overline l}
\newcommand \gb {\overline g}

\newcommand \Rt {\widetilde R}

\newcommand \gt {\widetilde g}

\newcommand \Gt {\widetilde G}

\newcommand \del \partial

\newcommand \nablat {\widetilde{\nabla}}

\newcommand \Ric {\text{Ric}}

\newcommand \Acal   {\mathcal A}

\newcommand \RR         {\mathbb R}

\newcommand \ee         {\end{equation}}
\newcommand \la \langle
\newcommand \ra \rangle

\newcommand \rhoh {\widehat{\rho}}

\usepackage{libertine}
\usepackage[libertine]{newtxmath}

\usepackage{microtype}

\begin{document}

\title{ 	Characteristic First-Order Structure of  f(R) Gravity
\\
 in Spherical Symmetry}

\author{Philippe G. LeFloch\footnote{
% (Corresponding Author)
Laboratoire Jacques-Louis Lions \& Centre National de la Recherche Scientifique,
Sorbonne Universit\'e, 4 Place Jussieu, 75252 Paris, France.
Email: \emph{contact@philippelefloch.org}
\newline $^\dag$ Centre of Mathematical Analysis, Geometry and Dynamical Systems, Instituto Superior T\'ecnico,  Universidade de Lisboa, Av. Rovisco Pais, 1049--001 Lisbon, Portugal. Email: \emph{filipecmena@tecnico.ulisboa.pt}
\newline
$^\flat$ Centre of Mathematics, Universidade do Minho, Campus de Gualtar, 4710-057 Braga, Portugal.
\newline
% --- Preferred (MSC 2020) ---
} 
\hskip.14cm
and 
Filipe C. Mena$^\dag$$^\flat$ }

\date{}
%April 2026}
%  (second revision)}

\maketitle
 
%===========================================================================

\begin{abstract} 
We develop { {an augmented characteristic, first-order formulation of}} the field equations in \(f(R)\) gravity governing the { {global}} evolution of a (possibly) massive scalar field $\phi$ under spherical symmetry.
{  This formulation is designed to isolate the genuine dynamical degrees of freedom while preserving the geometric structure of the theory.}
By treating the spacetime scalar curvature  
as an independent unknown, we obtain a \emph{closed} first-order nonlocal system for the pair $(\phi,R)$.
{  This augmentation eliminates the higher-derivative character of the original equations at the level of the principal part.}
Our formulation allows us to pose the characteristic initial value problem and to establish several { {structural}} properties of solutions. More precisely, we work in generalized Bondi--Sachs coordinates and prescribe initial data on an {  asymptotically flat}, future light cone with vertex at the center of symmetry, and we identify the { {minimal}} regularity conditions required at the center.
{  These regularity conditions are shown to be precisely those ensuring equivalence between the reduced system and the full $f(R)$ equations.}
Extending Christodoulou's method for the Einstein--scalar-field system, we recast the \(f(R)\) field equations as an integro-differential system of two coupled, first-order, nonlocal, nonlinear hyperbolic equations,
whose principal unknowns are the scalar field and the spacetime scalar curvature.
In deriving this reduced two-equation system, we impose natural assumptions on the scalar-field potential and on the function $f(R)$ governing the gravitational Lagrangian density.
{  These hypotheses correspond to standard viability and positivity conditions commonly imposed in the $f(R)$ literature.}
As an application, we prove several equivalence and monotonicity properties, including for the 
Hawking mass in this setting. 
The proposed formulation { {separates the essential null evolution from the radial constraint reconstruction} (via explicit integral relations)} on the future domain of dependence of the initial light cone.
{  This structure makes the system directly amenable to characteristic energy estimates and to stable numerical implementation in spherical symmetry.}
 \end{abstract}

%================================================================================
 
 \newpage 
  
{
 
\setcounter{secnumdepth}{2} 
\setcounter{tocdepth}{2}
\tableofcontents

} 

\clearpage 

\section{Introduction}
\label{section---1}

\subsection{Purpose of this paper}

\paragraph{Theories of modified gravity.}

In recent years, new observational data suggest that extensions of Einstein's field equations may be relevant for explaining the accelerated expansion of the Universe and certain galactic-scale instabilities. Many proposed theories of modified gravity nevertheless exhibit undesirable features, including
 { {a}} lack of strong hyperbolicity in commonly used gauges, { {possible}} non-uniqueness of solutions, { {or}} changes in the character of the evolution equations that may alter the system's causal structure. The development of robust numerical simulations in nonlinear, physically relevant regimes is hampered by the lack of mathematically rigorous formulations. 
Indeed, only a limited number of works address the formulation and well-posedness of these theories; see \cite{Avalos,Lehner2,Lehner,Cocke,Hilditch1,Felice,Figueras,LeFlochMa17a,Reall,Reall2,Pretorius,Salgado}. 

{  More recently, alternative, well-posed theories have been proposed in which the field equations are constructed so as to admit a strongly hyperbolic first-order reduction, thereby ensuring a well-defined causal structure and making nonlinear evolution amenable to both analytical and numerical investigation; cf.~\cite{Brady-Figueras,Figueras} (and the references therein), as well as \cite{LeFlochMa17a}--\cite{LeFlochMa26} for the near-Minkowski regime (without symmetry restriction).
 }

Among the various extensions of general relativity, $f(R)$ gravity is widely regarded as a natural and physically viable alternative to Einstein's theory \cite{Felice,Ferreira}. The $f(R)$ field equations are considerably more involved than the Einstein equations: in addition to the second-order Ricci curvature terms, they contain \emph{fourth-order} metric derivatives, {  in particular through second derivatives of the scalar curvature}. This partly explains why, despite their physical importance, rigorous mathematical results to date encompass only local well-posedness~\cite{Felice,LeFlochMa17a} and global nonlinear stability in the near-Minkowski regime \cite{LeFlochMa23}. { {We refer to the standard reviews for additional physical background and for the scalar--tensor (Einstein-frame) viewpoint; see, for instance, \cite{SotiriouFaraoni,Felice}.}}

{  The present work is primarily concerned with a mathematically robust characteristic formulation in spherical symmetry; we therefore restrict ourselves to structural assumptions ensuring well-defined null evolution and geometric control, while deferring model calibration and phenomenological comparisons to future studies.}

%-----------------------------------------------------

\paragraph{Evolution in spherical symmetry.}

In this paper, we initiate the study of the global evolution of a (possibly massive) scalar field in $f(R)$ gravity { {in}} spherical symmetry. Christodoulou \cite{Chr,Chr2,Chr3} developed, within general relativity, a framework to analyze a massless scalar field in spherical symmetry using Bondi-type coordinates. This framework was also employed to study black-hole formation by numerical methods; cf.~\cite{Goldwirth-Piran} and the references therein. {  
In that context, a first-order characteristic formulation was crucial for separating constraint equations from true evolution equations along outgoing null directions. This structure allows one to evolve the dynamical degrees of freedom by solving transport-type equations along characteristics, while reconstructing the remaining metric coefficients by radial integration. Such a separation is essential both for establishing well-posedness and for implementing stable numerical schemes.}

{  These developments inspired further investigations of gravitational collapse, in particular because the characteristic first-order formulation makes the causal propagation of matter fields transparent and provides direct access to geometric quantities such as trapped surfaces and the Hawking mass along null hypersurfaces.}
We refer to~\cite{Brady-etal,Hilditch2,Zhang-Lu,Michel-Moss}, as well as additional mathematical results; cf.~\cite{Chae,Costa-Mena,Costa-Duarte-Mena} (classical solutions) and \cite{LeFloch-Mena} (generalized solutions). Building on Christodoulou's subsequent work using a double-null foliation~\cite{Chr-bv,Chr-naked}, 
{  many studies have also addressed gravitational collapse toward black holes in spherical symmetry, largely exploiting the fact that a first-order null formulation allows one to track horizon formation and mass monotonicity directly at the level of the evolution system.}
Next, with the growing interest in modified gravity, subsequent works include \cite{Quo-Joshi, Zhang-etal-2016, Chow}, which mostly consider quadratic curvature corrections to the Einstein--Hilbert action. Despite these advances, rigorous mathematical results on the \emph{global geometry} of Cauchy developments in $f(R)$ gravity are still lacking. 
{  A central difficulty is that the fourth-order character of the $f(R)$ equations obscures the underlying hyperbolic structure and complicates the identification of a minimal set of evolution variables. Without a suitable first-order reformulation, both analytical energy estimates and characteristic numerical evolution become significantly more involved.}

{  We use a (possibly massive) scalar field as a standard test matter model in spherical symmetry: it is flexible enough to include the massless Christodoulou regime while allowing physically motivated self-interactions, and our analysis relies only on mild sign/defocusing assumptions on $U$.}

%----------------------------------------------------------------------------------------- 

\paragraph{Aim of this paper.}
	
Our aim is to advance the mathematical analysis of $f(R)$ gravity by building on the work of Christodoulou \cite{Chr,Chr2,Chr3}. 
We also { {draw on}} the formulation proposed by LeFloch and Ma \cite{LeFlochMa17a,LeFlochMa23} for the study of the equations of $f(R)$ gravity in the near-Minkowski regime. In the present paper, we focus on spacetimes containing a (possibly massive) scalar field evolving in spherical symmetry and, drawing on Christodoulou's insights, we provide a framework suited to addressing the global causal structure of such spacetimes. 

{  The key structural step is the construction of a genuine first-order characteristic system in which the higher-derivative terms of $f(R)$ gravity are absorbed into an augmented variable. This reduction makes the principal part explicitly hyperbolic along null directions, isolates the true dynamical degrees of freedom, and renders the constraint propagation transparent.}

Our companion paper \cite{LeFloch-Mena} considered the evolution of a self-gravitating \emph{massive scalar field} and provided a first step toward understanding the role of a massive scalar field in the \emph{global} spacetime geometry. While in \cite{LeFloch-Mena} we considered generalized solutions, in the present paper we consider solutions of class $C^1$ that satisfy regularity conditions at the center and are {  asymptotically flat}. Following Christodoulou, we use a generalization of the Bondi--Sachs coordinates and formulate the characteristic initial value problem with data imposed on a future light cone.
{  These coordinates are the natural setting for global null evolution in spherical symmetry and are precisely the framework in which one can separate null propagation from radial constraint reconstruction.}

The $f(R)$ system is significantly more involved than the Einstein system, and a key structure is uncovered here by introducing an augmented formulation in which the metric and its scalar curvature are regarded as independent unknowns. We thus obtain an integro-differential system of two coupled, first-order, nonlinear hyperbolic equations whose principal unknowns are the spacetime scalar curvature and, in this case, the scalar field. We then prove consistency and regularity properties for this system and its solutions, as well as monotonicity properties of the Hawking mass in the $f(R)$ setting. In the process, we formally recover the general relativity limit (when ${f(R) \to R}$), as well as the special case of a massless scalar field. 

{ {To better highlight the novelty and the purpose of our approach, we emphasize the following point. While Bondi-type coordinates and conformal variables are classical tools, their combination in the $f(R)$ setting does \emph{not} automatically yield a usable first-order characteristic reduction, due to the higher-derivative terms and the need to control the sign of the metric coefficient recovered from the radial constraints. Our main contribution is an \emph{augmented characteristic gauge} in which the scalar curvature---or, equivalently, the conformal variable $\rho=\kappa^{-1}\log f'(R)$ under the standard viability condition $f'(R)>0$---is treated as an independent unknown. { Here, $\kappa>0$ is a normalization (bookkeeping) parameter in the definition of $\rho$; in applications one may consider a family $f_\kappa$ such that $f_\kappa(R)\to R$ (equivalently $f_\kappa'(R)\to 1$, hence $\rho\to 0$) in the Einstein--Hilbert limit. (Cf.~also \eqref{equakdkdj}, below.) We denote by $\phi$ the scalar field under consideration.}

This choice closes the evolution at first order and leads to a \emph{closed, first-order, nonlocal} system for the pair $(\phi,R)$ (equivalently $(\phi,\rho)$), which we prove to be \emph{equivalent to the full $f(R)$ field equations} within a mild $C^1$ regularity class at the center. The resulting formulation separates the essential null evolution from the radial constraint reconstruction (via explicit integral relations), is directly compatible with characteristic energy estimates, and is designed to provide a robust starting point for characteristic numerical evolution in spherical symmetry.}}
{  In particular, the first-order structure obtained here is tailored to permit characteristic energy estimates and to ensure that no hidden higher-order derivatives reappear in the evolution subsystem, a property that is essential for both stability analysis and numerical implementation.}

{  For background on characteristic evolution and its numerical implementation, we refer to the review \cite{Winicour}. We do not pursue observational fitting here; rather, we aim to provide a formulation suitable for subsequent numerical and phenomenological studies once a specific $f(R)$ model and parameter regime are selected. In particular, the Hawking mass monotonicity derived below provides a built-in geometric diagnostic along the evolution.}

%--------------------------------------------------------------------------------------------------------------------------------------------

\subsection{A first-order formulation of $f(R)$ gravity}

\paragraph{Action of $f(R)$ modified gravity.}

Recall first that Einstein's theory is based on the Einstein--Hilbert action
\be
\Acal_{\text{EH}}[\phi, g] := \int_{\Mcal} \Big( {R_g  \over 16 \pi} + L[\phi,g] \Big) \, dV_g,
\ee
associated with a $(1+3)$--dimensional spacetime $({\Mcal},g)$ with signature $(-, +,+,+)$.
Hence, the functional $\Acal_{\text{EH}}[\phi, g] $ is determined by the scalar curvature $R_g$ of the metric $g$, which represents the geometry of spacetime, and by a Lagrangian $L[\phi,g]$ describing the matter content and represented by certain fields $\phi$ defined on $\Mcal$. Here, $dV_g$ denotes the canonical volume form associated with $g$.
As is well-known, the action $\Acal_{\text{EH}}[\phi, g]$ is (formally) critical at metrics $g$
satisfying Einstein's field equations
\be
\label{Eq1-01}
G := \Ric - {R_g \over 2} \, g = 8 \pi \, T[\phi,g],
\ee
in which the right-hand side 
\bel{Eq:12}
T_{\alpha\beta}[\phi,g] := -2 \, {\delta L \over \delta g^{\alpha\beta}} [\phi,g]  + g_{\alpha\beta}\, L[\phi,g]
\ee
is the stress-energy tensor of the matter system { {(Greek indices $\alpha, \beta,\dots=0,\ldots,3$ denoting spacetime components)}}.
Our emphasis in this paper is on massive scalar fields $\phi$ described by 
\be
\label{stress-energy}
T_{\alpha\beta} = \nabla_\alpha \phi \nabla_\beta \phi- \Big( {1 \over 2} \nabla^\gamma \phi \nabla_\gamma \phi + U(\phi) \Big) g_{\alpha\beta}, 
\ee
in which the potential function is $U=U(\phi)$, so that the field $\phi$ satisfies the Klein--Gordon equation 
\be
\Box_g \phi = U'(\phi).
\end{equation}
We study an extension of Einstein's theory, defined as follows. A function $f: \RR \to \RR$ being given, we consider the modified gravity action for the $f(R)$ theory
\be
\Acal_{\text{MG}}[\phi,g] =: \int_{\Mcal} \Big( {f(R_g) \over 16 \pi} + L[\phi, g]\Big) \, dV_g,
\ee
whose critical points satisfy the following \emph{field equations of modified gravity}
\bel{Eq1-14}
E_{\alpha\beta} :=
f'(R_g) \, G_{\alpha\beta} - \frac{1}{2} \Big( f(R_g) - R_g f'(R_g) \Big) g_{\alpha\beta}
+  \big( g_{\alpha\beta} \, \Box_g   - \nabla_\alpha \nabla_\beta \big)  \big( f'(R_g) \big) 
= 8 \pi \, T_{\alpha\beta}[\phi,g]. 
\ee
The right-hand side is still given by \eqref{Eq:12}. The modified gravity tensor $E_{ab}$ replaces
the Einstein tensor $G_{ab}$ and satisfies $\nabla^a E_{ab} = 0$, 
so that $T_{ab}$ is also divergence-free. 

%--------------------------------------------------------------------

\paragraph{Main result.}

In local coordinates, \eqref{Eq1-14} consists of a nonlinear system of fourth-order partial differential equations, while Einstein's theory in \eqref{Eq1-01} leads to second-order equations. We investigate
how to incorporate the effect of these higher-derivative terms into techniques developed earlier for the Einstein equations.
The condition $f'(R_g) >0$ in \eqref{hypo-Einstein}, below, is fundamental throughout the $f(R)$ theory and, in particular, it allows us to introduce the \emph{conformal metric} $\gt$ and the \emph{augmented variable} $\rho$ (in \eqref{equa-233}, below):  
\bel{equakdkdj}
\gt_{\alpha\beta} := e^{\kappa \rho} g_{\alpha\beta}, 
\qquad 
\rho := {1 \over \kappa} \log f'(R_g). 
\ee 
{  Recall that $\kappa>0$ is a normalization (bookkeeping) parameter in the definition of $\rho$; in applications one may consider a family $f_\kappa$ such that $f_\kappa(R)\to R$ (equivalently $f_\kappa'(R)\to 1$, hence $\rho\to 0$) in the Einstein--Hilbert limit.} 

{  Under the hypotheses below and within a mild $C^{1}$ regularity class at the center, we derive an \emph{augmented characteristic gauge} in which the full $f(R)$--scalar-field system reduces to a \emph{closed, first-order, nonlocal} integro-differential evolution for the pair $(\phi,R_g)$ (equivalently $(\phi,\rho)$). The remaining metric coefficients are then recovered by radial integration of constraint equations, in the spirit of Christodoulou's characteristic formulation for the Einstein--massless scalar field system.}

We summarize our results as follows. 

\begin{maintheorem}[Structure of $f(R)$ gravity in spherical symmetry]
\label{main-theo}
Consider the field equations in $f(R)$ gravity \eqref{Eq1-14}, coupled to a 
possibly massive, real-valued scalar field $\phi$. Suppose that the defining function $f= f(R)$ and the matter potential $U=U(\phi)$ satisfy the following conditions: 
\bel{hypo-Einstein}
\aligned 
& (1) 
\, &  \phi \, U'(\phi) &\geq 0, 
\\
& (2) 
\, &  U(\phi) &\geq 0,   
\\
& (3) 
\, & f'(R) &>0, 
\\
& (4)  
\, &  
f(R) &\leq R \, f'(R).  
\endaligned
\ee 
Then, the field equations of $f(R)$ gravity for spherically symmetric spacetimes can be reduced to an \emph{integro-differential system} consisting of two first-order coupled, nonlinear hyperbolic equations, stated in Proposition~\ref{main-propo-first-order}, below. 
Indeed, this system is equivalent to the full { {$f(R)$ field equations}} in the class of $C^1$ solutions that are suitably regular at the center in the sense of Definition~\ref{definition}, below. Furthermore, the Hawking mass is future non-decreasing along radial directions and non-increasing along null directions. In the \emph{formal} limit ${f(R) \to R}$ and 
$U(\phi)\to 0$, one recovers Christodoulou's formulation of the Einstein--massless scalar field system. 
\end{maintheorem}

\vskip.15cm

\paragraph{Comments on the results.}

We point out that our assumptions are quite natural since they ensure positivity and monotonicity properties that also arise in the massless case. The condition $\phi \, U'(\phi) \geq 0$ is imposed since it guarantees that the Klein--Gordon energy is \emph{defocusing}, so that the forward evolution will not be limited by the matter model. 
The conditions ${U(\phi) \geq 0}$ and ${f(R) \leq R \, f'(R)}$ are required to prove that the Hawking mass is non-negative. 

{  These conditions are not chosen to privilege a specific phenomenological model; rather, they are \emph{structural hypotheses} ensuring that (i) the characteristic reduction closes at first order, (ii) the metric reconstruction remains globally meaningful, and (iii) the Hawking-mass mechanism yields the required sign/monotonicity properties.}

{ {To further motivate our structural conditions  \eqref{hypo-Einstein}, we record here  the following remarks.}
\begin{itemize}
\item { {The condition $f'(R)>0$ is standard in $f(R)$ gravity. It ensures that the effective gravitational coupling does not change sign and, more fundamentally, it is precisely the condition under which the model is dynamically equivalent to a scalar--tensor theory in the Einstein frame. This equivalence motivates the introduction of the conformal metric and of the variable $\rho=\kappa^{-1}\log f'(R)$; see the review \cite{SotiriouFaraoni}.}}

\item { {The inequality $R f'(R)-f(R)\ge 0$ 
is used in our characteristic formulation to control the sign of the integrand in the radial constraint that reconstructs $e^{\nu-\lambda}$, and therefore to prevent breakdown of the characteristic gauge. It is also the condition that yields the non-negativity of the Hawking mass, a central geometric quantity in spherical symmetry.}}

\item { {The conditions $U(\phi)\ge 0$ and $\phi\,U'(\phi)\ge 0$ are standard defocusing/positivity assumptions for the scalar field. They ensure that the matter model does not introduce negative-energy mechanisms obstructing global control and allow us to recover the familiar monotonicity properties from the Einstein--(massless) scalar-field setting.}}
\end{itemize}
}

{ {It is also useful to relate the present characteristic formulation to the standard scalar--tensor (Einstein-frame) representation of $f(R)$ gravity, valid under $f'(R)>0$, in which the additional degree of freedom can be interpreted as a scalar field coupled to Einstein gravity; cf.\ the reviews \cite{SotiriouFaraoni,Felice}. Within this viewpoint, treating $R$ (equivalently $\rho$) as an independent variable is precisely what allows one to close the characteristic reduction at first order while retaining a formulation directly adapted to null evolution and to geometric monotonicity arguments. Finally, our hypotheses encompass concrete models used in cosmology and gravitational physics (in suitable curvature/parameter regimes), including the Starobinsky model $f(R)=R+\alpha R^{2}$ \cite{Starobinsky80} and the Hu--Sawicki model \cite{HuSawicki}.}}

\begin{remark} 
  
Assumptions (3)--(4) are consistent with the structural conditions under which Hawking-type horizon topology theorems extend to $f(R)$ gravity. 
Under the viability conditions $f'(R)>0$ and $f''(R)\ge 0$, the theory is dynamically equivalent to a scalar--tensor system satisfying an effective null energy condition. 
In this setting, cross-sections of stationary event horizons in four spacetime dimensions must have spherical topology; see, for instance, Sec.~2 of \cite{Namp}. 
Thus, the structural hypotheses adopted here exclude exotic topological black holes and are consistent with astrophysical expectations.
\end{remark}

\subsection{Outline of this paper}

In Section~\ref{section---2}, we introduce Bondi coordinates and express the field equations of $f(R)$ gravity for a scalar field. We then identify the essential equations, which imply the full set of $f(R)$ field equations, provided a mild regularity condition is assumed at the center. In Section~\ref{section---A2}, we analyze the regularity at the center. In Section~\ref{section---3}, we introduce a first-order formulation and establish that it is equivalent to the full set of $f(R)$ equations. Finally, in Section~\ref{section--- 4} we study of the Hawking mass and establish monotonicity properties. We summarize our conclusions and indicate perspectives in the final Section~\ref{sec:conclusions}.

%==================================================================================

\section{Field equations of $f(R)$ gravity in Bondi coordinates}
\label{section---2}

\subsection{Spacetime metric}
\label{defined metric}

{ {Our first task is to present the equations of $f(R)$ gravity when spherical symmetry is assumed. We use {
Greek indices $\alpha,\beta\ldots \in \{0,1,2,3\}$ and 
Latin indices $a,b,\ldots \in \{0,1,2,3\}$}
for coordinate and frame indices, respectively. }}
 The functions under consideration are assumed to be sufficiently regular \emph{away} from the center of symmetry, so that all first-order derivatives\footnote{$C^1$ regularity is sufficient, provided second-order derivatives are understood in the distributional sense.} are defined in the classical sense. The regularity at the center will be discussed in the course of our analysis. 

We consider spherically symmetric spacetimes denoted by $(\Mcal, g)$ and satisfying the equations of $f(R)$-gravity. Following the physical and mathematical literature (cf.~the review given in~\cite{LeFlochMa17a,LeFlochMa23} and the references therein), we work with the \emph{conformal metric}
\bel{equa-233} 
\gt_{\alpha\beta} := e^{\kappa \rho} g_{\alpha\beta}, 
\qquad 
\rho := {1 \over \kappa} \log f'(R) 
\ee
rather than with the physical metric $g$. { {In the present section we work in a regime in which $f'(R)>0$ and the map $R\mapsto \kappa^{-1}\log f'(R)$ is (locally) invertible, so that $\rho$ can be used as an independent variable.}} 
Here, $\kappa$ is a positive parameter such that 
${f(R) \to R}$ when ${\kappa \to 0}$. 
Therefore, it is the metric $\gt$ that we express in the Bondi form advocated by Christodoulou \cite{Chr}.  We denote by $u \geq 0$ a (future) null coordinate and $r>0$ the radial variable associated with the area radius of the orbits of symmetry, while $g_{S^2} = {d\theta^2 + \sin^2 \theta d\varphi^2}$ denotes the canonical metric on the standard $2$-sphere $S^2$. Hence, in these generalized Bondi coordinates $(u,r)$ we write 
\bel{eq:metricrad}
\gt = -e^{2\nu}du^2 - 2e^{\nu+\lambda} dudr + r^2 g_{S^2},
\ee
where $\nu= \nu(u,r)$ and $\lambda= \lambda(u,r)$ depend only on the variables $(u,r)$. 

In view of \eqref{eq:metricrad}, we define a null frame $(e_0, e_1, e_2, e_3) =(n,l,\zeta_1,\zeta_2)$  by setting 
\bse
\bel{frame1}
\aligned
& e_0 = n := e^{- \nu}\del_u - \frac{1}{2}e^{- \lambda}\del_r,
\qquad
&& e_1 = l := e^{- \lambda}\del_r,
\\
& e_2 = \zeta_1 := r^{-1}\del_\theta,
\qquad
&& e_3 = \zeta_2 := \frac{1}{r\sin\theta}\del_{\varphi}. 
\endaligned
\ee
By construction, $(\zeta_1,\zeta_2)$ is a (locally defined) orthonormal frame on the sphere $S^2$, while $n, l$ are null vectors satisfying
$\gt(n,l) = -1$. 
The corresponding co-frame denoted $(\omega^0, \omega^1, \omega^2, \omega^3)$ is easily computed to be 
\bel{coframe}
\aligned
& \omega^0 = e^{\nu}du,
\qquad
\omega^1 = e^{\lambda}dr + \frac{1}{2}e^{\nu}du,
\qquad
\omega^2 = rd\theta,
\qquad
\omega^3 = r\sin\theta d\varphi.
\endaligned
\ee 
\ese 
In addition, the contravariant form of the conformal metric reads
\bel{eq;492}
\gt^{\alpha\beta} = -n^\alpha l^\beta - l^\alpha n^\beta + \zeta_1^\alpha \zeta_1^\beta
+ \zeta_2^\alpha \zeta_2^\beta.
\ee 
Observe that the integral curves of the vector fields $n$ and $l$ represent incoming and outgoing light rays in the spacetime $(\Mcal, \gt)$, respectively. We will use the \emph{light ray operator}
\be
Dw := e^{\nu} n(w) = \del_uw - \frac{1}{2}e^{\nu- \lambda}\del_r w
\ee
and the \emph{wave operator} acting on a spherically symmetric function $w$, namely 
\bel{Bondi d'Alembert}
\aligned
\Box_{\gt} w 
& =- \frac{2}{re^{\nu+\lambda}}\Big(D(\del_r(rw))
-{r \over 2} \del_r(e^{\nu- \lambda})\del_rw\Big). 
\endaligned
\ee

%-----------------------------------------------------------------------------------------------------------------------------------

\subsection{Derivation of the field equations of interest}

\paragraph{Frame components.}

It is convenient to introduce the inverse of the function $R \mapsto \frac{1}{\kappa}\ln f'(R)$,
which we denote by $\rho \mapsto \Rstar(\rho)$. {  Indeed, $R^*(\rho)$ is the inverse function satisfying}
\be
\label{fprime}
f'( R_\star(\rho) ) = e^{\kappa\rho}.
\end{equation}
We easily check that $\rho$ satisfies the wave equation
\bel{equa-Laplace-rho} 
\Box_{\gt}\rho = \frac{2}{\kappa}\Big(\Wstar(\rho) - \frac{4\pi}{3}e^{-2\kappa\rho}\big(\sigma+4 \, U(\phi)\big)\Big), 
\ee
where we introduced the \emph{first potential function} 
\be
\label{asym W_star}
\Wstar(\rho) := \frac{2f(\Rstar(\rho)) - \Rstar(\rho)e^{\kappa\rho}}{6e^{2\kappa\rho}}.
\end{equation}
The modified gravity tensor $E_{ab}$ defined in \eqref{Eq1-14}, in frame components, takes the form
\bel{eq fR components}
\aligned
& E_{nn} = \frac{e^{\kappa\rho}}{re^{\nu+\lambda}}\Big(\frac{1}{2}e^{\nu- \lambda}\del_r(\nu+\lambda) - 2\del_u\lambda\Big)
- \frac{3\kappa^{2}e^{\kappa\rho}}{2e^{2\nu}}(D\rho)^2,
\\
& E_{nl} = - \frac{e^{\kappa\rho}}{r^2e^{\nu+\lambda}}\big(\del_r(re^{\nu- \lambda})-e^{\nu+\lambda}\big) - \frac{\Vstar(\rho)}{2e^{\kappa\rho}},
\\
& E_{ll} = e^{\kappa\rho-2\lambda}\Big(\frac{2}{r}\del_r(\nu+\lambda) - \frac{3}{2}\kappa^{2}|\del_r\rho|^2\Big),
\\
& E_{\zeta_1\zeta_1} = E_{\zeta_2\zeta_2}=e^{\kappa\rho-2\lambda}\Big(\frac{1}{r}\del_r(\nu- \lambda)+\del_{rr}\nu+\del_r\nu(\del_r\nu- \del_r\lambda)-e^{- \nu+\lambda}\del_{ur}(\nu+\lambda)\Big)\\
& \quad\quad\quad\quad \quad \quad- \frac{3}{2}\kappa^2e^{\kappa\rho- \nu- \lambda}D\rho\del_r\rho+e^{- \kappa\rho}\Vstar(\rho),
\endaligned
\ee
where we have introduced the {\em second potential function} associated with $f$:
\bel{eq:pot} 
\Vstar(\rho) := {1 \over 2}  \Rstar(\rho) f'(\Rstar(\rho)) - {1 \over 2} f(\Rstar(\rho)).
 \ee 
On the other hand, the frame components of the stress-energy tensor \eqref{stress-energy} are 
\bel{eq Bondi components T}
\aligned 
& T_{nn} = e^{-2\nu}(D\phi)^2,
\hspace{3.8cm}
T_{ll} = e^{-2\lambda}(\del_r\phi)^2,
\\
& T_{\zeta_{1}\zeta_{1}} = T_{\zeta_{2}\zeta_{2}} = -e^{- \kappa\rho}\left(\frac{\sigma}{2}+U(\phi)\right),\qquad
 T_{nl}=T_{ln}=e^{- \kappa\rho}U(\phi),
\endaligned
\ee
in which the notation
\bel{trace-of-grad-phi} 
\sigma := \widetilde g^{\alpha\beta}\del_\alpha\phi\del_{\beta}\phi
= -2e^{- \nu- \lambda+\kappa\rho}\del_r\phi D\phi
\ee
stands for the (Lorentzian) norm of the scalar field gradient with respect to the conformal metric. In particular, the trace of the stress-energy tensor $T$ 
equals
\bel{traceT}
\text{tr}(T) = - \sigma- 4 \, U(\phi).
\ee

%--------------------------------------------------------------------------------------------------------------------

\paragraph{Field equations in Bondi coordinates.}

By combining the equations above, we are in a position to write down the field equations of $f(R)$ gravity coupled with a scalar field. 
Two of the essential field equations are associated with the components $(n,l)$ and $(l,l)$, that is,
\bel{eq:522}
\aligned
& E_{nl} = 8\pi e^{- \kappa\rho}U(\phi),
\qquad \quad  E_{ll} = 8\pi e^{-2\lambda}|\del_r\phi|^2,
\endaligned
\ee 
which are equivalent to the following first-order differential equations: 
\bel{equa-three-equa}
\aligned
\del_r(\nu+\lambda) & =  \frac{3}{4}\kappa^{2}r|\del_r\rho|^2 + 4\pi r e^{- \kappa\rho} |\del_r \phi|^2,
\\
\del_r\big(re^{\nu- \lambda}\big) & =  \Big(1 - r^{2}e^{-2\kappa\rho} \big( \Vstar(\rho) + 8\pi U(\phi) \big)\Big)e^{\nu+\lambda}. 
\endaligned
\ee
From the equation $\Box_g \phi = U'(\phi)$, { {or equivalently $e^{\kappa\rho}\Box_{\gt}\phi- \kappa\,\gt\!\left(\nablat\rho,\nablat\phi\right) =U'(\phi)$}},  
we find
\bel{Bondi scalar}
D \Big( \del_r(r\phi) \Big) = {r \over 2} \del_r(e^{\nu- \lambda})\del_r\phi +
\frac{\kappa}{2}r\big(D\rho \del_r\phi + D\phi \del_r\rho\big)-{r \over 2} e^{\nu+\lambda- \kappa\rho}U'(\phi).
\ee
Furthermore, we need an evolution equation for the scalar curvature $R= \Rstar(\rho)$. By taking the trace of \eqref{Eq1-14}, we deduce the scalar equation
\bel{trace}
\text{tr}(E) = f'(R)R-2f(R)+3\Box_{g}f'(R) =8\pi\text{tr}(T),
\ee
which yields (using~\eqref{Bondi d'Alembert}, \eqref{traceT}, and \eqref{block-g})
\bel{eq:599}
\aligned
D\big(\del_r(r\rho)\big)
= 
& \frac{e^{\nu+\lambda}}{2}\Big(1 - e^{-2\lambda} - r^2e^{-2\kappa\rho}\big(\Vstar(\rho)+8\pi U(\phi)\big)\Big)\del_r\rho 
\\
& - {8 \pi r\over 3 \kappa} \, e^{- \kappa\rho} \del_r\phi D\phi
- \frac{1}{\kappa}re^{\nu+\lambda}
\Big(
\Wstar(\rho) - \frac{16}{3}\pi e^{-2\kappa\rho}U(\phi) \Big).
\endaligned
\ee
Here, we have also used \eqref{trace-of-grad-phi} and \eqref{asym W_star}. 

%-------------------------------------------------------------------------------------------------------------------------------------------

\subsection{Essential field equations in Bondi coordinates} 

We will now show that, if regularity at the center is imposed, then the full set of equations  is equivalent to the system formed by the equations $E_{ll}=8\pi T_{ll}$, $E_{nl}=8\pi T_{nl}$, the trace of \eqref{Eq1-14} and the Klein-Gordon equation. 

\begin{proposition}[Essential field equations of $f(R)$ gravity]
\label{prop essential f(R) Bondi}
Consider the system of $f(R)$-modified gravity coupled to a scalar field $\phi$ and assume spherical symmetry. When the conformal metric $\gt$
is expressed in Bondi coordinates as stated in~\eqref{eq:metricrad}:
\begin{enumerate}

\item The field equations \eqref{Eq1-14}
are equivalent to
 
\bei

\item two differential equations for the metric coefficients $\nu,\lambda$, namely \eqref{equa-three-equa}, and 

\item a transport equation for the matter field $\phi$, namely \eqref{Bondi scalar}. 

\eei  
\noindent These equations are referred to as the \emph{essential field equations} of $f(R)$-modified gravity in spherical symmetry. { {(Together with the trace equation \eqref{eq:599}, they form a closed characteristic system once $\rho$ is treated as an independent variable.)}}

\item  Here,  the function $\rho = \frac{1}{\kappa}\ln f'(R)$ is a nonlinear function
of the scalar curvature $R$ of the physical metric $g$ and is determined from up to second-order derivatives of $\nu$ and $\lambda$.
In fact,  $\rho$ also satisfies the transport equation \eqref{eq:599}. 

\item  More precisely, if the spacetime satisfies the regularity condition at the center: 
\bel{equa-key-regular} 
\lim_{r \to 0} r^2 \, e^{-\kappa\rho} \, \big(E - 8\pi T\big)(D,D) = 0. 
\ee
\end{enumerate} 
If the essential equations \eqref{equa-three-equa}, \eqref{Bondi scalar}, and \eqref{eq:599} hold
(which includes the $(n,l)$ and $(l,l)$ components in the null frame), 
then the remaining equations corresponding to the components $(n,n)$ and $(\zeta_a, \zeta_a)$ are automatically satisfied.
\end{proposition}

{ {We point out that the quantity appearing in \eqref{equa-key-regular}, more explicitly, reads}} 
\be
\aligned
 r^2 e^{\kappa \rho} 
 \Bigg(
 \frac{e^{\kappa\rho}}{re^{\nu+\lambda}} \, \frac{1}{2}e^{\nu- \lambda} \Big( \frac{3}{4}\kappa^{2}r|\del_r\rho|^2 + 4\pi r e^{- \kappa\rho} |\del_r \phi|^2 \Big)
 - 2  \frac{e^{\kappa\rho}}{re^{\nu+\lambda}} \del_u\lambda
- \frac{3\kappa^{2}e^{\kappa\rho}}{2e^{2\nu}}(D\rho)^2
- e^{-2\nu}(D\phi)^2 \Bigg) 
\endaligned
\ee
and is thus equivalent to a condition on the behavior of the geometric and matter unknowns. 

\begin{proof} 
\bse
 We need to establish that the remaining equations
\be
E_{nn} = 8\pi T_{nn}, \qquad
E_{\zeta_1\zeta_1} = 8\pi T_{\zeta_1\zeta_1}, \qquad
E_{\zeta_2\zeta_2} = 8\pi T_{\zeta_2\zeta_2}
\ee
can be deduced from the essential equations, as follows. We introduce the tensor
$
F_{ab} := E_{ab} - 8\pi T_{ab}$ and 
we prove first that $F_{\zeta_1\zeta_1}=F_{\zeta_2\zeta_2}= 0$. In fact, by a straightforward calculation
\be
\aligned
F_{\zeta_1\zeta_1}& =  e^{\kappa\rho}\Big(\frac{1}{r}e^{-2\lambda}\del_r(\nu- \lambda)+e^{-2\lambda}(\del_{rr}\nu+(\del_r\nu)^2- \del_r\nu\del_r\lambda)-e^{- \nu- \lambda}\del_{ur}(\nu+\lambda)\Big)\\
& \quad - \frac{3}{2}\kappa^2e^{\kappa\rho- \nu- \lambda}D\rho\del_r\rho+e^{- \kappa\rho}\Big(\Vstar(\rho)+8\pi \big(\frac{\sigma}{2}+U(\phi)\big)\Big).
\endaligned
\ee
\ese
\bse
Plugging the expression above into \eqref{eq:Rtilde}, we obtain 
\be
\aligned
e^{\kappa\rho}\Rt+2F_{\zeta_1\zeta_1} 
=
& -3\kappa^2e^{\kappa\rho- \nu- \lambda}D\rho\del_r\rho+2e^{- \kappa\rho}\Big(\Vstar(\rho)+8\pi \big(\frac{\sigma}{2}+U(\phi)\big)\Big)
\\
&  +\frac{2}{r^{2}}e^{\kappa\rho- \nu- \lambda}\Big(e^{\nu+\lambda}- \del_r\big(re^{\nu- \lambda}\big)\Big).
\endaligned
\ee
Therefore, it follows from \eqref{equa-three-equa} that
\bel{R and Fzeta}
e^{\kappa\rho}\Rt= -3\kappa^2e^{\kappa\rho- \nu- \lambda}D\rho\del_r\rho+4e^{- \kappa\rho}\Big(\Vstar(\rho)+2\pi \big(\sigma+4 \, U(\phi)\big)\Big)-2F_{\zeta_1\zeta_1},
\ee 
which, using the conformal transformation of the scalar curvature  
\bel{R-conformal-transf}
\aligned
R & = e^{\kappa\rho}\big(\widetilde R + 3\kappa\Box_{\gt}\rho - \frac{3}{2}\kappa^{2}\gt(\nablat\rho,\nablat\rho)\big)
\\
& = e^{\kappa\rho}\big(\Rt + 3\kappa\Box_{\gt}\rho + 3\kappa^{2}e^{- \nu- \lambda}D\rho\del_r\rho\big), 
\endaligned
\ee
yields
\bel{Fzeta= 0}
\aligned
2F_{\zeta_1\zeta_1}=
& \quad -R+4e^{- \kappa\rho}\Big(\Vstar(\rho)+2\pi \big(\sigma+4 \, U(\phi)\big)\Big)+3\kappa e^{\kappa\rho}\Box_{\gt}\rho\\
& =  -6e^{\kappa\rho}\Big(\Wstar(\rho) - \frac{4\pi}{3}e^{-2\kappa\rho}\big(\sigma+4 \, U(\phi)\big)\Big)+3\kappa e^{\kappa\rho}\Box_{\gt}\rho.
\endaligned
\ee
This expression vanishes thanks to \eqref{equa-Laplace-rho} derived earlier and, by \eqref{Fzeta= 0}, we get 
$F_{\zeta_1\zeta_1}=F_{\zeta_2\zeta_2}= 0$, as claimed.
\ese
\bse
We will now treat $F_{nn}$. In the null frame, we have 
\be
g_{\alpha\beta} = \left(
\begin{array}{cccc}
 0 & -e^{- \kappa\rho} & 0 & 0 \\
-e^{- \kappa\rho} & 0  & 0 & 0 \\
 0 & 0  & e^{- \kappa\rho} & 0 \\
 0 & 0  & 0 & e^{- \kappa\rho}
\end{array} \right),
\ee
so from $\nabla^{a}F_{ab} = 0$ we obtain
\be
- \nabla_n F_{lc} - \nabla_l F_{nc} + \nabla_{\zeta_1}F_{\zeta_1c} + \nabla_{\zeta_2}F_{\zeta_2c} = 0.
\ee
Using
$
\nabla_{a}F_{bc} = \langle a, F_{bc}\rangle - \Gamma_{ab}^{c}F_{dc} - \Gamma_{ac}^{d}F_{bd},
$ 
and the fact that in spherical symmetry $\del_{\zeta_1} F_{\zeta_1c} = \del_{\zeta_2} F_{\zeta_2c} = 0$,
we rewrite the above relation as
\bel{proof 5.1 2}
\aligned
& \quad - \langle n, F_{lc}\rangle +\Gamma_{nc}^{d}F_{ld} +\Gamma_{nl}^d F_{dc}
- \langle l, F_{nc}\rangle +\Gamma_{lc}^{d}F_{nd}\\ & \quad + \Gamma_{ln}^{d}F_{dc}
 - \Gamma_{\zeta_1c}^{d}F_{\zeta_1d} - \Gamma_{\zeta_1\zeta_1}^{d}F_{dc}
 - \Gamma_{\zeta_2c}^{d}F_{\zeta_2d} - \Gamma_{\zeta_2\zeta_2}^{d}F_{dc} = 0.
\endaligned
\ee 
Making the choice $c = n$ in \eqref{proof 5.1 2} and using \eqref{connection-on-frame}, we get
\be
\aligned
& - \langle n, F_{ln}\rangle + \Gamma_{nl}^{d}F_{d n} + \Gamma_{nn}^{d}F_{ld} - \langle l,F_{nn}\rangle + \Gamma_{ln}^{d}F_{d n} + \Gamma_{ln}^{d}F_{n d}
\\
&  - \Gamma_{\zeta_1\zeta_1}^{d}F_{d n} - \Gamma_{\zeta_1 n}^{d}F_{\zeta_1d} - \Gamma_{\zeta_2\zeta_2}^{\delta}F_{d n} - \Gamma_{\zeta_2 n}^{d}F_{d\zeta_2} = 0.
\endaligned
\ee
So taking into account that
$
F_{\zeta_1\zeta_1} = F_{\zeta_2\zeta_2} = F_{ln} = {F_{ll} = 0}$, 
we find
\ese
\bse
\be
\del_r F_{nn} + 2\Big(\del_r\nu + \frac{1}{r} - \frac{\kappa}{2} \del_r\rho\Big)F_{nn} = 0.
\ee
Finally, by integrating this equation it follows that
\be
F_{nn}(u,r) = \frac{r_0^2}{r^2}e^{2\big((\nu- \frac{\kappa}{2}\rho)|_{(u,r_0)} - (\nu- \frac{\kappa}{2}\rho)|_{(u,r)}\big)}F_{nn}(u,r_0), 
\ee
in which $F_{nn} = e^{-2 \nu} F(D,D)$. 
By the assumption \eqref{equa-key-regular} on the regularity at the center, we conclude that  $F_{nn} = 0$ for all $r_0$ and this completes the proof.
\ese
\end{proof}

%==============================================================================

\section{Analysis of the regularity at the center}
\label{section---A2}
 
 \subsection{Augmented conformal formulation in Bondi coordinates}
 
{ 
\paragraph{Aim.}

One of our objectives to justify, at the level of the field equations, that the augmented conformal system provides a \emph{closed} characteristic formulation of $f(R)$ gravity without introducing any spurious degrees of freedom.
The difficulty is that the original $f(R)$ equations contain derivatives of the scalar curvature (hence third-order derivatives of the metric coefficients in Bondi gauge), so that a naive first-order characteristic reduction is not closed.
Our strategy is to treat the conformal variable $\rho=\kappa^{-1}\ln f'(R)$ as an \emph{independent} unknown (denoted $\widehat\rho$ in \eqref{equa-hdj4}, below) and to work with an augmented system in which $R$ is not eliminated a priori.
The core statement is a \emph{consistency mechanism}: under mild $C^1$ regularity and center regularity. the nonlinear constraint $f'(R)=e^{\kappa\widehat\rho}$ propagates and the augmented system turns out to be equivalent to the original one in the admissible class.
Finally, we will establish center regularity properties needed to make sense of $R$ and $\partial_rR$ and to justify the limiting argument as $r\to0$.
}

\paragraph{The augmented system.}

In the modified field equations, there are terms that contain derivatives of the scalar
curvature, that is, third order derivatives of the metric functions $\nu$ and $\lambda$. This leads to an essential difficulty in dealing with the equations.
Following LeFloch and Ma~\cite{LeFlochMa17a}, we introduce an \emph{augmented system}
 where the relation
\be
\rho = \frac{1}{\kappa}\ln f'(R)  \quad \text{(conformal formulation)} 
\ee
is no longer imposed but $\rho$ 
 is regarded as an independent variable.
To clarify the notation, we introduce a new independent variable denoted by 
\bel{equa-hdj4}
\rhoh \quad \text{(augmented conformal formulation)} 
\ee
which plays the role of $\rho$ and will {  coincide} with $\rho$ once the constraint below is enforced. 
In view of Proposition~\ref{prop essential f(R) Bondi}, this leads us to the following system 
\bel{Bondi system diff augmented1.5}
\aligned
\del_r(\nu+\lambda)
& =  \frac{3}{4}\kappa^{2}r|\del_r\rhoh|^2 + 4\pi r e^{- \kappa\rhoh} \, |\del_r\phi|^2,
\\
\del_r\big(re^{\nu- \lambda}\big)
 & =   \big(1 - r^{2}e^{-2\kappa\rhoh} \big( \Vstar(\rhoh) + 8\pi U(\phi) \big) \big)e^{\nu+\lambda},
 \\
D(\del_r(r\phi)) + {r \over 2} e^{\nu+\lambda- \kappa\rhoh}U'(\phi)
& =  \frac{1}{2}e^{\nu+\lambda}\Big(1 - e^{-2\lambda} - r^2e^{-2\kappa\rhoh}\big(\Vstar(\rhoh)+8\pi U(\phi)\big)\Big)\del_r\phi \\
& \quad + \frac{\kappa r}{2}\big(D\rhoh \del_r\phi + D\phi \del_r\rhoh\big), 
\\
D\big(\del_r(r\rhoh)\big) + {r \over \kappa}  e^{\nu+\lambda} \Wstar(\rhoh)
& =  \frac{1}{2}e^{\nu+\lambda}\Big(1 - e^{-2\lambda} - r^2e^{-2\kappa\rhoh}\big(\Vstar(\rhoh)+8\pi U(\phi)\big)\Big)\del_r\rhoh
\\
& - \frac{8\pi r}{3\kappa} e^{- \kappa\rhoh}\del_r\phi D\phi
  + {r \over \kappa} e^{\nu+\lambda} \frac{16\pi}{3}e^{-2\kappa\rhoh}U(\phi).
\endaligned
\ee
Observe that, as before, the last equation of \eqref{Bondi system diff augmented1.5} is equivalent to
\be
\label{box-tilde-rho}
\Box_{\gt}\rhoh  - \frac{2}{\kappa} \Wstar(\rhoh)
=  - {8\pi \over 3\kappa}e^{-2\kappa\rhoh}\big(\sigma+4 \, U(\phi)\big).
\ee
Now that we have defined the augmented system, the next question is whether a solution of \eqref{Bondi system diff augmented1.5} (with a certain regularity) is also a solution of the original system. We emphasize that to be a solution of the full set of field equations, a solution of  \eqref{equa-three-equa}, \eqref{Bondi scalar}, and \eqref{eq:599} 
must additionally satisfy the condition
\bel{Bondi system diff augmented1.5 constraint}
f'(R) = e^{\kappa\rhoh},
\ee
which we regard as a nonlinear differential constraint on the solutions.
 Conversely, any classical solution of the original $f(R)$ system with $\rho=\frac{1}{\kappa}\ln f'(R)$ satisfies \eqref{Bondi system diff augmented1.5} upon setting $\rhoh=\rho$; thus the two formulations are equivalent in the admissible class once \eqref{Bondi system diff augmented1.5 constraint} holds on the initial cone (and is then propagated by the evolution), {  as we will show.}
 
 %------------------------------------------------------------------------------------------------------------------------ 

\subsection{Regularity and consistency properties} 

{ 

The augmented system allows us to view $R$ and $\widehat\rho$ as \emph{a priori} independent.
To recover the original theory, we must show that the quantity $R-\Rstar(\widehat\rho)$ vanishes identically once it vanishes at the center.
We are going to prove that $R-\Rstar(\widehat\rho)$ satisfies a first-order radial ODE with an explicit integrating factor determined by $\widehat\rho$ and the Bondi gauge.
Center regularity then forces the corresponding integration constant to be $0$, which yields $R=\Rstar(\widehat\rho)$ and hence $f'(R)=e^{\kappa\widehat\rho}$ throughout the domain.
}

We start by stating the main result of this section.  

\begin{proposition}[Consistency property  for the augmented conformal system]
\label{prop Bondi augemented}
Let $(\phi,\rhoh,\nu,\lambda)$ be a $C^1$ solution of the augmented conformal system 
\eqref{Bondi system diff augmented1.5} defined on $[0,u_0]\times [0,\infty)$.
Then, provided that the functions 
\be
\del_r(r\phi) \text { and } \del_r(r\widehat\rho) \text{ are } C^1 \text{ on } [0,u_0]\times [0,+ \infty),
\ee
one has
\be
f'(R) =e^{\kappa\rhoh}.
\ee 
Conversely, if a classical solution of the original $f(R)$ system satisfies $f'(R)=e^{\kappa\rho}$ on the initial cone and $\partial_r(r\phi),\partial_r(r\rho)\in C^1$, then setting $\rhoh=\rho$ yields a $C^1$ solution of \eqref{Bondi system diff augmented1.5}. Hence, within the admissible class, the augmented and original formulations are equivalent.
\end{proposition}

The proof is based on the following two lemmas, whose proof is postponed to Appendix \ref{append-twolemmas}. 

\begin{lemma}
\label{sec 6 lem R_g apriori}
Consider a $C^1$ solution $(\phi,\rhoh,\nu,\lambda)$ of \eqref{Bondi system diff augmented1.5} defined on $[0,u_0]\times [0,\infty)$. Assume that $\phi$ and $\rhoh$ satisfy the regularity condition at the center:
\be
\del_r(r\phi),\,\del_r(r\rhoh) \in C^1\big([0,u_0]\times[0,\infty)\big).
\ee
Then, the scalar curvature $R$ is well defined on $[0,u_0]\times(0,\infty)$ and $\del_r R$ is also well defined on $[0,u_0]\times(0,\infty)$ and continuous.
Moreover, $\del_{rr}(\nu+\lambda)$ and $\del_{ur}(\nu+\lambda)$ can be continuously extended on $[0,u_0]\times[0,+ \infty)$.
\end{lemma}

\begin{lemma}
\label{sec 6 lem R_g main estimate}
Let $(\phi,\rhoh,\nu,\lambda)$ be a $C^1$ solution of \eqref{Bondi system diff augmented1.5} defined on $[0,u_0]\times[0,\infty)$. Then, for any $r,r_0 > 0$, the following identity holds:  
\bel{ODE2}
{ {S(u,r)= S(u,r_0)\exp\!\Big(- {\int_{r_{0}}^r \frac{2}{s}e^{- \kappa\rhoh(u,s)} \, ds}\Big),}}
\ee
in which
\bel{eq:939k}
\aligned 
&F_{\zeta_1\zeta_1}^{\rhoh}(u,r):= e^{\kappa\rhoh-2\lambda}\Big(\frac{1}{r }\del_r(\nu- \lambda)+\del_{rr}\nu+(\del_r\nu)^2- \del_r\nu\del_r\lambda-e^{\lambda- \nu}\del_{ur}(\nu+\lambda)\Big)\\
& \qquad\qquad\qquad- \frac{3}{2}\kappa^2e^{\kappa\rhoh- \nu- \lambda}D\rhoh\del_r\rhoh+e^{- \kappa\rhoh}\Big(\Vstar(\rhoh)+8\pi \big(\frac{\sigma}{2}+U(\phi)\big)\Big),\\
& { {S(u,r)  := (R - \Rstar(\rhoh))(u,r)\, \exp\!\big(e^{-\kappa \rhoh(u,r)}\big).}}
\endaligned
\ee
{ {In particular, one has the identity $R-\Rstar(\rhoh)=2F_{\zeta_1\zeta_1}^{\rhoh}$ and, therefore, 
\be
S(u,r)=2F_{\zeta_1\zeta_1}^{\rhoh}(u,r)\exp(e^{-\kappa\rhoh(u,r)}).
\ee
}}
\end{lemma} 

\begin{proof}[Proof of Proposition  \ref{prop Bondi augemented}] 
\bse
We observe, first, that for any given $r_0$, \eqref{ODE2} may be written as
\be
\aligned 
\frac12\,S(u,r)\,e^{-e^{-\kappa\rhoh(u,r)}} 
& =F^{\rhoh}_{\zeta_1\zeta_1}(u,r)
\\
& 
=F^{\rhoh}_{\zeta_1\zeta_1}(u,r_0)\,e^{\,e^{-\kappa\rhoh(u,r_0)}-e^{-\kappa\rhoh(u,r)}}
\exp\!\Big(- \int_{r_0}^r{\frac{2}{s}e^{- \kappa\rhoh(u,s)} \, ds}\Big).
\endaligned
\ee
Therefore, to show that $R= \Rstar(\rhoh)$, it suffices to show that fixing $(u,r) \in [0,u_0]\times (r_0,+ \infty),$
\bel{atcenter.1}
\lim_{r_{0}\to 0}\Big(F^{\rhoh}_{\zeta_1\zeta_1}(u,r_0)e^{- \int_{r_0}^r{\frac{2}{s}e^{- \kappa\rhoh} \, ds}}\Big) = 0.
\ee
In fact, since $\rhoh\in C^{1}\big([0,u_0]\times[0,+ \infty)\big)$, we have (for $r$ fixed)
\bel{center.1}
e^{- \int_{r_0}^r{\frac{2}{s}e^{- \kappa\rhoh} \, ds}}= O\!\big(r_0^{\,\alpha}\big),
\qquad \alpha:=2e^{-\kappa\rhoh(u,0)}>0,\quad \text{as } r_{0}\to 0.
\ee
On the other hand, since $\rhoh,\phi\in C^{1}\big([0,u_0]\times[0,+ \infty)\big)$, the same holds for $\nu$ and $\lambda$, hence 
\bel{Fzeta}
\aligned
F_{\zeta_1\zeta_1}^{\rhoh}(u,r_0)  
& =  O\Big(1+\big|\frac{\del_r(\nu- \lambda)}{r_0}\big|+|\del_{rr}\nu|+|\del_{ur}(\nu+\lambda)|\Big), \qquad \mbox{as $r_0\to 0$}.
\endaligned
\ee
Hence, all we have to do is to check that $\del_r(\nu- \lambda)(u,r_0),\ r_0\del_{rr}\nu(u,r_0)$ and $r_0\del_{ur}(\nu+\lambda)(u,r_0)$ tend to $0$ as $r_0\to 0$. 
\ese

\bse
We begin with the term $r_0\del_{ur}(\nu+\lambda)$: In fact, since from Lemma \ref{sec 6 lem R_g apriori}, $\del_{ur}(\nu+\lambda)$ can be continuously extended on $[0,u_0]\times [0,+ \infty)$, we obtain
\be
r_0\del_{ur}(\nu+\lambda)(u,r_0)\to 0,\qquad\mbox{as $r_0\to0$}.
\ee
We next consider the term $\del_r(\nu- \lambda)$: From the first two equations of \eqref{Bondi system diff augmented1.5}, we find 
\be
\Big(e^{\nu- \lambda}\del_r(\nu- \lambda)\Big)(u,0) = \Big(e^{\nu+\lambda}\del_r(\nu+\lambda)\Big)(u,0) = 0,
\ee
so $\del_r(\nu- \lambda)(u,0) = 0$, as claimed.
\ese

\bse
Now, we discuss the term $r_0\del_{rr}\nu(u,r_0)$ and observe that 
\be
\aligned
e^{\nu- \lambda}\del_r\nu& =  \frac{1}{2r}\del_r\big(re^{\nu- \lambda}\big)+\frac{e^{\nu- \lambda}}{2}\del_r(\nu+\lambda)- \frac{e^{\nu- \lambda}}{2r}\\
& =  \frac{e^{\nu+\lambda}}{2r}\Big(1-r^2e^{-2\kappa\rhoh}\big(\Vstar(\rhoh)+8\pi U(\phi)\big)\Big)+re^{\nu- \lambda}\Big(\frac{3}{8}\kappa^{2}|\del_{r}\rhoh|^{2}+\frac{2\pi }{e^{\kappa\rhoh}}|\del_r\phi|^{2}\Big)- \frac{e^{\nu- \lambda}}{2r}
\endaligned
\ee
from the first two equations of \eqref{Bondi system diff augmented1.5}. By a straightforward calculation we have 
\be
\aligned
& \big(r\del_{rr}\nu\big)(u,r_0) 
\\
& = \Big[e^{2\lambda}\del_r\lambda-e^{- \nu+\lambda}\Big(\frac{e^{\nu+\lambda}-e^{\nu- \lambda}}{2r}\Big)-{r \over 2} \del_r\Big(re^{2\lambda}\big(\Vstar(\rhoh)+8\pi U(\phi)\big)\Big)+{r \over 2} \del_{rr}(\nu+\lambda)\Big](u,r_0).
\endaligned
\ee
Since $\rhoh,\phi,\nu,\lambda\in C^1$, and since $\del_{rr}(\nu+\lambda)$ can be continuously extended on $[0,u_0]\times [0,+ \infty)$ as proven in Lemma \ref{sec 6 lem R_g apriori}, we get that 
\bel{r del rr of nu}
r_0\del_{rr}\nu(u,r_0) = O\Big(|\del_r\lambda|(u,r_0)+\big|\frac{e^{\nu+\lambda}-e^{\nu- \lambda}}{r_0}\big|\Big)\quad\mbox{as $r_0\to0$}.
\ee
In the previous arguments, we have shown that 
$\del_r(\nu+\lambda)(u,0) = \del_r(\nu- \lambda)(u,0) = 0$,
which is equivalent to saying 
\be
\del_r\nu(u,0) = \del_r\lambda(u,0) = 0.
\ee
Moreover, combining this result with the second equation of \eqref{Bondi system diff augmented1.5}, we find
\be
\aligned
& \frac{\big(e^{\nu+\lambda}-e^{\nu- \lambda}\big)(u,r_0)}{r_0} 
\\
& = \Big[re^{-2\kappa\rhoh+\nu+\lambda}\big(\Vstar(\rhoh)+8\pi U(\phi)\big)+e^{\nu- \lambda}\del_r(\nu- \lambda)\Big](u,r_0) \to  0,\quad \mbox{as $r_0 \to 0$}.
\endaligned
\ee
Taking these facts into \eqref{r del rr of nu}, we obtain 
$r_0\del_{rr}\nu(u,r_0) \to  0$
as $r_0\to 0$. Together with Lemma~\ref{sec 6 lem R_g main estimate} and \eqref{R and Fzeta}, this yields $R=\Rstar(\rhoh)$ and hence $f'(R)=e^{\kappa\rhoh}$. 
\ese
\end{proof}

%=================================================================================
	
\section{A first-order formulation of $f(R)$ gravity}
\label{section---3}
 
\subsection{The reduced system} 
\label{system-of-interest}

We now present a first-order reduction of the characteristic initial value problem on $[0,u_0]\times [0,+ \infty)$, and discuss the regularity class of the solutions. Data are prescribed on a light cone labelled $u= 0$, and suitable asymptotic flatness conditions are imposed at spacelike infinity $r \to + \infty$. We prescribe initial data on the hypersurface $\big\{ u= 0 \big\}$ as: 
\be
\aligned
& \phi|_{u= 0} = \phi_0,
\qquad
\rho|_{u= 0} = \rho_0,
\endaligned
\ee
for some given functions $\phi_0$ and $\rho_0$. The metric on the initial slice is completely determined by $\phi_0$ and $\rho_0$, thanks to the two equations \eqref{eq:403}. 

\begin{definition}  
\label{definition}
The \emph{first regularity condition at the center\footnote{In view of \eqref{eq fR components} and \eqref{eq Bondi components T}, this condition can be expressed in terms of first-order derivatives of the unknowns.}} is defined as 
\be
\label{regularity} 
\lim_{r\to0} r^2e^{-\kappa\rho}\big(E - 8\pi T\big)(D,D) = 0.
\ee
The \emph{second regularity condition at the center} is defined as 
\be
\label{regularity2}
\del_r(r\phi) \text{ and }\del_r(r\rho) \text{ are } C^1\text{ on } [0,u_0]\times [0,+ \infty).
\ee
\end{definition}

Moreover, we require that the spacetime is \emph{  asymptotically flat}, that is, 
\be
\lim_{r \to +\infty} (\nu+\lambda) = 0, 
\ee 
which implies 
\be
\lim_{r \to +\infty} (\nu, \lambda) = 0. 
\ee

Under the regularity conditions at the center and the asymptotic flatness condition, we can reduce the field equations 
{  to the following integro-differential system }
with the main unknowns $\phi$ and $\rho$:  
\bel{notresysteme}
\aligned
D(\del_r(r\phi))
& =  \frac{1}{2}e^{\nu+\lambda}\Big(1 - e^{-2\lambda} - r^2e^{-2\kappa\rho}\big(\Vstar(\rho)+8\pi U(\phi)\big)\Big)\del_r\phi
\\
 & \quad + \frac{\kappa r}{2}\big(D\rho \del_r\phi + D\phi \del_r\rho\big)-{r \over 2} e^{\nu+\lambda- \kappa\rho}U'(\phi), 
\\
D\big(\del_r(r\rho)\big)
& =  \frac{1}{2}e^{\nu+\lambda}\Big(1 - e^{-2\lambda} - r^2e^{-2\kappa\rho}\big(\Vstar(\rho)+8\pi U(\phi)\big)\Big)\del_r\rho
\\
& \quad - \frac{8\pi r}{3\kappa} e^{- \kappa\rho}\del_r\phi D\phi- {r \over \kappa} e^{\nu+\lambda}\Big(\Wstar(\rho) - \frac{16\pi}{3}e^{-2\kappa\rho}U(\phi)\Big),
\endaligned
\ee
in which the metric coefficients $\nu, \lambda$ are given by 
\bel{eq:403}
\aligned
\nu + \lambda
& =  - \frac{3}{4}\kappa^{2} \int_r^{+ \infty}s|\del_r\rho|^2ds - 4\pi\int_r^{+ \infty} e^{- \kappa\rho} \, s|\del_r\phi|^2 \, ds,
\\
e^{\nu- \lambda}
& =  {1 \over r}  \int_0^r\Big(1 - s^{2}e^{-2\kappa\rho} \big( \Vstar(\rho) + 8\pi U(\phi) \big)\Big)e^{\nu+\lambda} \, ds. 
\endaligned
\ee
We refer to \eqref{notresysteme} as the \emph{augmented conformal system} of modified gravity.
 
\begin{definition}\label{define C1 type solution}
A pair $(\phi,\,\rho)$ is called an \emph{admissible solution} of the system \eqref{notresysteme} if $(\phi,\rho) \in C^1$ and satisfies the first and second regularity conditions at the center.
\end{definition}

%----------------------------------------------------------------------------------------------------------------------------

\subsection{First-order version of the augmented conformal system}

{  
To make contact with the dynamical content of the theory and to facilitate characteristic evolution, we now recast the augmented conformal system in a genuinely first-order form. 
The higher-derivative structure of $f(R)$ gravity obscures the propagating degrees of freedom, and a first-order reduction is essential to identify the physical null evolution and the associated constraint hierarchy. 
In the spirit of Bondi-type formulations, we introduce variables adapted to radial transport along outgoing light rays, so that the dynamical subsystem and the metric reconstruction can be clearly separated. 
This reformulation makes the causal propagation of both the scalar field and the additional scalar degree of freedom manifest, and it prepares the system for stable characteristic numerical implementation.  
}

Indeed, inspired by Christodoulou's framework \cite{Chr}, we find it convenient (both for theoretical and numerical applications) to reformulate the system \eqref{notresysteme} by introducing the first-order variables 
\be
\aligned
l:= \del_r(r\rho),
\qquad
h:= \del_r(r\phi).
\endaligned
\ee
Given any function $f=f(u,r)$, we denote by $\overline f$ the function obtained by averaging over the interval $[0,r]$, as follows: 
\be
\label{mean-value-operator}
\overline f(u,r) = \frac{1}{r}\int_0^r f(u,s) \, ds. 
\ee
In view of the regularity conditions $\lim_{r \to 0} \del_r(r\phi) = \lim_{r \to 0} \del_r(r\rho) = 0$, we deduce that 
\be
h = \del_r(r\hb),\qquad l= \del_r(r\lb),
 \qquad  \del_r \hb= \frac{h- \hb}{r},\qquad \del_r \lb= \frac{l- \lb}{r},
\ee
as well as
\be
\hb= \phi, \qquad \quad \lb= \rho.
\ee
It is also convenient to define
\be
g:= \Big(1-r^2e^{-2\kappa\lb}\big(\Vstar(\lb)+8\pi U(\hb)\big)\Big)e^{\nu+\lambda},
\ee
so that, by \eqref{eq:403}, 
\be
\gb=e^{\nu- \lambda}. 
\ee

With this notation, from \eqref{notresysteme}-\eqref{eq:403} we obtain 
\be
\aligned
\label{Bondi system int augmented2a}
& Dh = \frac{1}{2r}\big(g- \gb\big)(h- \hb) +
\frac{\kappa}{2}\big(D\lb(h- \hb) + D\hb(l- \lb)\big)-{r \over 2} e^{\nu+\lambda- \kappa\lb}U'(\hb),
\\
& \kappa Dl = \frac{\kappa }{2r}\big(g- \gb\big)(l- \lb)- \frac{8\pi}{3}e^{- \kappa\lb}(h- \hb) D\hb - re^{\nu+\lambda}\Big(\Wstar(\lb) - \frac{16\pi}{3}e^{-2\kappa\lb}U(\hb)\Big), 
\endaligned
\ee
together with
\be
\aligned 
& \nu + \lambda = - \frac{3}{4} \kappa^2 \int_r^{+ \infty}\frac{|l- \lb|^2}{s} \, ds - 4\pi\int_r^{+ \infty}\frac{|h- \hb|^2}{se^{\kappa\lb}} \, ds.
\endaligned
\ee 
It remains to express $D \hb$ and $D \lb$ arising in \eqref{Bondi system int augmented2a} in terms of $l$ and $h$ and lower-order contributions. We introduce two auxiliary variables 
\be
\label{variables-pq}
p:= D\hb, \qquad 
q:= D\lb, 
\end{equation}
which, in view of the regularity conditions, satisfy
\be
\lim_{r \to 0} (rp) = 0,
\qquad
\lim_{r \to 0} (rq) = 0.
\ee
We then arrive at one of the conclusions announced in Theorem \ref{main-theo}. 

\begin{proposition}[A first-order formulation of $f(R)$ gravity]
\label{main-propo-first-order}
\bse\label{Bondi system int augmented3} 
The equations of $f(R)$ gravity take the form of a first-order system, in which $h$ and $l$ are the main unknowns: 
\be
\aligned
& Dh = \frac{1}{2r}\big(g- \gb\big)(h- \hb) +
\frac{\kappa}{2}\big(q(h- \hb) + p(l- \lb)\big)-{r \over 2} e^{\nu+\lambda- \kappa\lb}U'(\hb),
\\
& Dl = \frac{1}{2r}\big(g- \gb\big)(l- \lb)
- \frac{8\pi}{3\kappa}e^{- \kappa\lb}(h- \hb) p-{r \over \kappa} e^{\nu+\lambda}\Big(\Wstar(\lb)- \frac{16\pi}{3}e^{-2\kappa\lb}U(\hb)\Big),
\endaligned
\ee
in which the metric coefficient $\nu + \lambda$ is computed by 
\be
\aligned 
& \nu + \lambda = - \frac{3\kappa^{2}}{4}\int_r^{+ \infty}\frac{|l- \lb|^2}{s} \, ds - 4\pi\int_r^{+ \infty}\frac{|h- \hb|^2}{se^{\kappa\lb}} \, ds,
\endaligned
\ee
and the auxiliary variables $p$ and $q$ are defined by the following integral expressions:  
\bel{p,q}
\aligned 
&p = \frac{\kappa}{2r}\int_0^r\Big((l- \lb)p + (h- \hb)q\Big) \, ds +\frac{1}{2r}\int_0^r\frac{e^{\nu- \lambda}(h- \hb)}{s} \, ds
- \frac{1}{2r}\int_0^rse^{\nu+\lambda- \kappa\lb}U'(\hb) \, ds, 
\\
& \kappa q 
= - \frac{8\pi}{3 r}\int_0^r\frac{(h- \hb)p}{e^{\kappa\lb}} \, ds+ \frac{\kappa}{2r}\int_0^r{\frac{e^{\nu- \lambda}(l- \lb)}{s} \, ds}
 - \frac{1}{r}\int_0^r se^{\nu+\lambda}\Big (\Wstar(\lb)-
\frac{16\pi}{3} e^{-2\kappa\lb}U(\hb)\Big) \, ds. 
\endaligned
\ee
\ese
\end{proposition}

From the proposed formulation, by taking $\kappa= 0$ and $U=0$, we recover the formulation of the Einstein-massless scalar field system discovered by Christodoulou \cite{Chr, Chr2}.

\begin{proof}  
\bse
Using the identity 
\be
r\del_r\Big(- \frac{1}{2}e^{\nu- \lambda}\del_r\phi\Big) = - \frac{1}{2}r\del_r\big(e^{\nu- \lambda}\big)\del_r\phi
- \frac{1}{2}re^{\nu- \lambda}\del_{rr}\phi,
\ee
we deduce 
\be
r\del_r D\hb + D\hb = Dh - \frac{1}{2}\del_r\big(e^{\nu- \lambda}\big)(h- \hb) + \frac{1}{2r}e^{\nu- \lambda}(h- \hb).
\ee
This, together with the first equation in \eqref{Bondi system int augmented2a}, yields 
\bel{Bondi augmentation2 proof 0} 
\del_r D\hb +\frac{1}{r}D\hb = \frac{\kappa}{2r}\big(D\lb(h- \hb) + D\hb(l- \lb)\big) + \frac{1}{2r^2}e^{\nu- \lambda}(h- \hb)- \frac{1}{2}e^{\nu+\lambda- \kappa\lb}U'(\hb).
\ee
Similarly, we derive 
\be
\del_rD\lb + \frac{1}{r}D\lb =   \frac{8\pi}{3\kappa r}e^{- \kappa\lb}(h- \hb)D\hb
+ \frac{1}{2r^2}e^{\nu- \lambda}(l- \lb)
- \frac{e^{\nu+\lambda}}{\kappa}\Big(\Wstar(\lb)- \frac{16\pi}{3}e^{-2\kappa\lb}U(\hb)\Big). 
\ee
Consequently, the system \eqref{Bondi system int augmented2a} is equivalent to  
\be
\label{Bondi-system-new}
\aligned
& \del_r(rp) - \frac{\kappa}{2}\Big((l- \lb)p +(h- \hb)q\Big) = \frac{1}{2r}e^{\nu- \lambda}(h- \hb)-{r \over 2} e^{\nu+\lambda- \kappa\lb}U'(\hb),
\\
& \del_r(rq) - \frac{8\pi(h- \hb)}{3\kappa e^{\kappa\lb}}p= \frac{e^{\nu- \lambda}}{2r}(l- \lb) - \frac{re^{\nu+\lambda}}{\kappa}\Big(\Wstar(\lb) - \frac{16\pi}{3} e^{-2\kappa\lb}U(\hb)\Big), 
\endaligned
\end{equation}
and we conclude by integrating these equations. 
\ese
\end{proof} 

%----------------------------------------------------------------------------------------------------------------------------------------------

\subsection{Consistency and regularity properties}

We now discuss the regularity of solutions to the first-order augmented conformal system \eqref{Bondi system int augmented3}. In particular, we establish the consistency of the regularity conditions, showing that the auxiliary variables coincide with the corresponding directional derivatives. Conversely, the integral definitions of the auxiliaries propagate these identities once they hold at the center. 

\begin{lemma}
Let $(h,l,\lambda,\nu,p,q)$ be a classical solution of \eqref{Bondi system int augmented3} defined on $[0,u_0]\times [0,+ \infty)$. Then, one has 
\be
\aligned
p & = D\hb,
\qquad
q  = D\lb.
\endaligned
\ee
\end{lemma}

\begin{proof}  
Under the regularity conditions, \eqref{Bondi augmentation2 proof 0} holds. Then, we have 
\be
r\del_r D\hb + D\hb = \frac{\kappa}{2}\big(q(h- \hb) + p(l- \lb)\big) + \frac{1}{2r}e^{\nu- \lambda}(h- \hb) 
- \frac{1}{2}r e^{\nu+\lambda- \kappa \lb} U'(\hb) .
\ee
{ 
Comparing this identity with the first identity in \eqref{p,q} and observing that (trivially) 
$\del_r\big(r D \hb \big) = \del_r(rp)$, we deduce that, by regularity, $ r D \hb - r p = c$, where $c$ is constant.} Then letting $r\to 0$ we conclude that
$c = 0$. For $q = D \lb$, a similar argument of proof applies.
\end{proof}

%---------------------------------------

Suppose now that $h$ and $l$ are solutions to the system \eqref{Bondi system int augmented3} defined on $[0,u_0]\times [0,+ \infty)$. If 
\bel{center.condition.2}
h,l\in C^1\big([0,u_0]\times [0,+ \infty)\big),
\ee
then the first and second regularity conditions at the center hold. In fact, it follows from the system \eqref{Bondi system int augmented3} that
\be
\label{more-conditions-center}
\hb,\,\lb,\,\nu,\,\lambda\in C^1\big([0,u_0]\times [0,+ \infty)\big).
\ee
In view of \eqref{eq fR components} and \eqref{eq Bondi components T}, we get
\be
\aligned
& r^2{e^{-\kappa\lb}}\big(E-8\pi T\big)(D,D) 
\\
& =  r{e^{\nu- \lambda}}\Big(\frac{e^{\nu- \lambda}}{2}\del_r(\nu+\lambda) - 2\del_u\lambda\Big)
- \frac{3r^2\kappa^{2}{ } }{2}(D\lb)^2-8\pi r^2{ }(D\hb)^2.
\endaligned
\ee
Hence, from the regularity conditions, and taking into account \eqref{more-conditions-center}, we deduce that
\be
\lim_{r\to0}{\Big(r^2e^{-\kappa\rho}\big(E-8\pi T\big)(D,D)\Big)}= 0, 
\ee
since $\nu,\lambda=O(r)$ and $D\hb,D\lb$ are bounded near $r=0$.  For the second regularity condition at the center, we compute  
\be
r\del_r\lb = l- \lb, \qquad  r\del_r\hb = h- \hb.
\ee  
 This shows that $h=\partial_r(r\phi)$ and $l=\partial_r(r\rho)$ are $C^1$ on $[0,u_0]\times[0,+\infty)$, as claimed.  Conversely, if the two center regularity conditions hold and $(h,l)\in C^0$, the right-hand sides of \eqref{Bondi system int augmented3} are continuous; integrating along characteristics then yields $(h,l)\in C^1$, so the two formulations are equivalent in the admissible class. 
 
%==============================================================================

\section{Hawking mass and monotonicity properties} 
\label{section--- 4}

\subsection{Hawking mass}

An important difference between our system \eqref{Bondi system int augmented3} and the one proposed by Christodoulou \cite{Chr} is that, in our setting, there is no \emph{ a priori} guarantee that the right-hand side of the representation for $e^{\nu- \lambda}$ in \eqref{eq:403} is non-negative. In order to control the sign of that expression and, in fact, to establish that its right-hand side remains globally positive if it is initially positive, we need some properties about the mass distribution of the spacetime. We thus consider the \emph{Hawking mass} $m =m(u,r)$, defined by 
\be
\label{mass-definition}
m:={r \over 2}  \Big( 1- \gt(\nablat r,\nablat r) \Big) 
= {r \over 2} \big(1-e^{-2\lambda}\big). 
\ee 

\begin{proposition}[Positivity of the Hawking mass in $f(R)$ gravity]
\label{prop mass positive}
Except in the case of Minkowski spacetime (i.e.\ when $m\equiv 0$), the mass function is \emph{positive} for all $r>0$ throughout the spacetime, and one has
\be
\label{positivity-of-mass}
1 - {2m \over r} = e^{-2\lambda}
= {1 \over r e^{\nu+\lambda}} 
\int_0^r \Big(
1 - s^2 e^{-2\kappa\lb} \big( \Vstar(\rho) + 8 \pi U(\phi) \big) 
\Big) 
\, e^{\nu+\lambda} \, ds
\in (0,1). 
\ee 
Equality $1-\tfrac{2m}{r}=1$ for all $r$ holds if and only if the spacetime is Minkowski. 
\end{proposition}

\begin{proof} 
\bse
Clearly, from \eqref{mass-definition} we see that $m < r/2$. For the mass to be positive we must prove that 
$
e^{\nu- \lambda} <  e^{\nu+\lambda}.
$
To do so, we observe that, by the first equation in \eqref{notresysteme}, 
\be
\del_r(e^{\nu+\lambda}) = e^{\nu+\lambda}\del_r(\nu+\lambda) \geq 0,
\ee
unless $\rho$ and $\phi$ are both constant in $r$. 
So that $e^{\nu+\lambda}$ is monotonically non-decreasing with respect to $r$. 
Since $\Vstar$ and $U$ are non-negative, the following inequality holds: 
\be
e^{\nu- \lambda} = \frac{1}{r}\int_0^r \big(1 - s^2e^{\kappa\rho}(\Vstar(\rho)+8\pi U(\phi)\big)e^{\nu+\lambda} ds 
\leq \frac{1}{r}\int_0^r e^{\nu+\lambda} \, ds \leq e^{\nu+\lambda},
\ee
where the first equality follows from the second equation in \eqref{notresysteme}.
This gives \eqref{positivity-of-mass} and guarantees the positivity of the mass.
\ese
M
\end{proof}

%------------------------------------------------------

\begin{proposition}[Monotonicity of the Hawking mass in $f(R)$ gravity. I]
\label{prop increasing m along r}
The mass is a non-decreasing function in radial directions:
\be
\label{mass-r-derivative}
\del_r m
= \frac{r^2e^{-2\kappa\lb}}{2}(\Vstar(\lb)+8\pi U(\hb)) + \frac{e^{-2\lambda}}{2}\Big( \frac{3\kappa^2}{4}|l- \lb|^2 + \frac{4\pi}{e^{\kappa\lb}} |h- \hb|^2\Big)
\geq 0. 
\ee
\end{proposition}

\begin{proof}  From the definition of $m$, we compute 
\be
\aligned
\del_r m & = \frac{1}{2}\Big(1- \frac{e^{\nu- \lambda}}{e^{\nu+\lambda}}\Big)
- {r \over 2} \del_r\Big(\frac{e^{\nu- \lambda}}{e^{\nu+\lambda}}\Big)
\\ 
& = \frac{1}{2} - \frac{1}{2e^{\nu+\lambda}}\del_r\big(re^{\nu- \lambda}\big)
+ \frac{re^{\nu- \lambda}}{2e^{\nu+\lambda}}\del_r(\nu+\lambda).
\endaligned
\ee
Then, by substituting the first and second equation of \eqref{notresysteme} and using the definitions of the previous section, we reach the desired identity. 
\end{proof}

%------------------------------------------------------------------------------------

Let us next consider the characteristic curves. By inspecting the system \eqref{Bondi system int augmented3}, we observe that the two evolution equations in the system share the same characteristic speed. Let $\chi=\chi(u,r_0)$ be the characteristic curve with initial data $r_0$, that is, 
\bel{eq_chara1}
\aligned
& \frac{d\chi}{du} = - \frac{1}{2}e^{\nu- \lambda},
\qquad
\quad
\chi(0) = r_0.
\endaligned
\ee
Consider the evolution of the mass along a characteristic curve. 

\begin{proposition}[Monotonicity of the Hawking mass in $f(R)$ gravity. II]
\label{prop decreasing of m along D}
The mass is a non-increasing function along the incoming light rays towards the future:
\be
\label{mass-D-derivative}
Dm = - \frac{r^2e^{- \nu- \lambda}}{4}\Big(3\kappa^2  |D\rho|^2 
          + 16\pi  e^{- \kappa\rho} |D\phi|^2 
          + (\Vstar(\rho) +8\pi U(\phi))e^{2\nu-2\kappa \rho}\Big) \leq 0.
\ee
\end{proposition}

\begin{proof}  
\bse
From \eqref{mass-definition} and \eqref{mass-r-derivative}, we have 
\be
\aligned
Dm  
 & = re^{-2\lambda}\del_u\lambda - \frac{r}{4}e^{\nu- \lambda}e^{-2\lambda}\del_r(\nu+\lambda) + \frac{1}{4}e^{-2\lambda}\big(\del_r(re^{\nu- \lambda})\big) - \frac{1}{4}e^{\nu- \lambda}
\\
 & = -{r \over 2} e^{\nu- \lambda}\big(-2e^{- \nu- \lambda}\del_u\lambda + \frac{1}{2}e^{-2\lambda}\del_r(\nu+\lambda)\big) + \frac{1}{4}e^{\nu- \lambda}\big(e^{- \nu- \lambda}\del_r(re^{\nu- \lambda}) - 1\big)
\\
 & = - \frac{r^2}{2}e^{\nu- \lambda}\Gt_{nn} + \frac{1}{4}e^{\nu- \lambda}\big(e^{- \nu- \lambda}\del_r(re^{\nu- \lambda})-1\big),
\endaligned
\ee
which implies 
\bel{eq Dm proof 1}
Dm = - \frac{r^2}{2}e^{\nu- \lambda}\Gt_{nn} - \frac{r^2}{4}e^{\nu- \lambda-2\kappa\rho}(\Vstar(R)+8\pi U(\phi)).
\ee
Recall that, in the proof of Proposition \ref{prop essential f(R) Bondi}, we found that $F_{nn} = 0$, that is, 
$
E_{nn} = 8\pi T_{nn}. 
$
Therefore, we have 
\be 
e^{\kappa\rho}\Gt_{nn} - \frac{3\kappa^2e^{\kappa\rho}}{2e^{2\nu}}|D\rho|^2 = 8\pi T_{nn}= \frac{8\pi |D\phi|^2}{e^{2\nu}}
\ee
and thus 
\be
\Gt_{nn} = \frac{3\kappa^2|D\rho|^2}{2e^{2\nu}} + \frac{8\pi|D\phi|^2}{e^{2\nu+2\rho}}.
\ee
Combining this result with \eqref{eq Dm proof 1}, we arrive at the desired identity. 
\ese
\end{proof}

%%----------------------------------------------------------------------

\subsection{  On the possibility of negative mass}

{  

The sign of the Hawking mass can be made completely transparent by
rewriting \eqref{positivity-of-mass} in the equivalent form
\begin{equation}
\label{mass-integral-form}
m(u,r)
=
\frac{1}{2e^{\nu+\lambda}(u,r)}
\int_0^r
s^2 e^{-2\kappa\rho}
\big(\Vstar(\rho)+8\pi U(\phi)\big)
\, e^{\nu+\lambda}\, ds.
\end{equation}
This formula shows that $m(u,r)$ is obtained by integrating an
\emph{effective energy density} over the interior of the sphere of
areal radius $r$, with a strictly positive weight factor
$e^{\nu+\lambda}$.

In spherical symmetry the Hawking mass coincides with the
Misner--Sharp mass, since
\[
1-\frac{2m}{r} = g(\nabla r,\nabla r),
\]
which is a purely geometric identity depending only on the
areal radius function and not on the specific gravitational
Lagrangian. Thus the notion of mass employed here is
geometrically canonical and does not depend on the choice of
field variables or on the conformal representation.

Under the structural assumptions
\[
\Vstar(\rho)\ge 0,
\qquad
U(\phi)\ge 0,
\]
the integrand in \eqref{mass-integral-form} is non-negative.
Consequently,
\[
m(u,r)\ge 0
\quad \text{for all } r>0,
\]
with equality if and only if the spacetime is Minkowski.

If $\Vstar$ is allowed to take negative values
---for instance in models violating $f''(R)\ge 0$---
then the integrand in \eqref{mass-integral-form}
may become negative in some region.
In that situation the effective scalar sector fails to satisfy
the null energy condition, and the Raychaudhuri-based
monotonicity argument breaks down.
Negative values of $m$ may then occur.

We emphasize that negativity of the
Hawking mass does \emph{not} indicate an inadequacy of the
Misner--Sharp/Hawking mass notion.
Rather, it signals that the chosen $f(R)$ model permits
negative effective energy densities.
Within the standard viability conditions adopted throughout
this paper, such pathologies are excluded and
$m(u,r)$ remains non-negative, preserving its usual
geometric and physical interpretation.
	
	}

\subsection{Invariant domain property}

As mentioned before, to guarantee that $(\phi,\,\rho)$ is a well-defined solution of the system \eqref{notresysteme}, it remains to be shown that the right-hand side of the expression for $e^{\nu- \lambda}$, that is, 
\be
\gb={1 \over r}  \int_0^r\Big(1 - s^{2}e^{-2\kappa\rho} \big( \Vstar(\rho) + 8\pi U(\phi) \big)\Big)e^{\nu+\lambda} \, ds, 
\ee
is \emph{non-negative}. We will use the mass function to check this property.

\begin{proposition}[Invariant domain property] 
\label{proposition gb>0}
Let $(\phi,\,\rho)$ be a $C^1-$type solution to the system \eqref{notresysteme} defined on $[0,u_0]\times [0,+ \infty)$. Assume that $\gb(0,r)>0$ on $[0,+ \infty)$. Then, $\gb(u,r)>0$ on $[0,u_0]\times [0,+ \infty)$.  
\end{proposition}

\begin{proof} 
\bse
We argue by contradiction. Assume that there exists $(\widehat{u},\widehat{r}) \in [0,u_0]\times [0,+ \infty)$ such that
\be\label{contradiction gb<0}
\gb(\widehat{u},\widehat{r})\le 0.
\end{equation}
We set
\be
u_*:= \sup\bigg\{u_1\in [0,u_0]\quad\text{such that}\quad \text{$\gb(u,r)>0$, for all $(u,r) \in [0,u_1]\times[0,\widehat{r}]$} \bigg\}.
\ee
Since $\gb(0,r)>0$ for all $r\in [0,\widehat{r}]$, we obtain  $0<u_*\le \widehat{u}$. By the definition of $u_*$, we have  
\begin{itemize}
\item $\gb>0$ on $[0,u_*)\times [0,\widehat{r}]$,
\item there exists $r_*\in [0,\widehat{r}]$ such that $\gb(u_*,r_*) = 0$.
\end{itemize}   
Next, on $[0,u_*)\times [0,r_*]$ we redefine 
\be
e^{\nu- \lambda}:= \gb,
\ee
so $(\phi,\,\rho)$ is a well-defined solution of the system \eqref{notresysteme}  on $[0,u_*)\times [0,r_*]$. Therefore, by Proposition \ref{prop decreasing of m along D}, we establish that 
\be
Dm = - \frac{r^2e^{- \nu- \lambda}}{4}\Big(3\kappa^2  |D\rho|^2 
+ 16\pi  e^{- \kappa\rho} |D\phi|^2 
+ (\Vstar(\rho) +8\pi U(\phi))e^{2\nu-2\kappa \rho}\Big)
\ee
for all $(u,r) \in [0,u_*)\times [0,r_*]$. Since the right-hand side of this identity is non-positive, we obtain
\be
\del_u m \le \frac{\gb}{2}\del_rm \quad \text{on} \quad [0,u_*)\times [0,r_*]
\ee	
{  Using $\partial_u m=re^{-2\lambda}\partial_u\lambda$ and \eqref{mass-r-derivative}, 
} for all $(u,r) \in [0,u_*)\times [0,r_*]$ we have 
\be
\del_u\lambda\le \frac{re^{\nu+\lambda-2\kappa\lb}}{4}(\Vstar(\rho)+8\pi U(\phi)) + \frac{re^{\nu- \lambda}}{4}\Big( \frac{3\kappa^2}{4}(\del_r\rho)^2 + \frac{4\pi}{e^{\kappa\lb}} (\del_r\phi)^2\Big).
\ee
\ese
\bse
It follows that for all $(u,r) \in [0,u_*)\times [0,r_*]$
\be
\lambda(u,r)\leq \lambda(0,r)+\int_0^u\left[\frac{re^{\nu+\lambda-2\kappa\rho}}{4}\big(\Vstar(\rho)+8\pi U(\phi)\big) 
+ \frac{re^{\nu- \lambda}}{4}\Big( \frac{3\kappa^2}{4}(\del_r\rho)^2 + \frac{4\pi}{e^{\kappa\rho}} (\del_r\phi)^2\Big)\right]du_1,
\ee
so $\lambda(u,r_*)$ remains bounded above as $u\uparrow u_*$. Hence $e^{-2\lambda(u,r_*)}$ stays bounded below by a { {positive}} constant on $[0,u_*)$.

On the other hand, since $\gb=e^{\nu-\lambda}$ and $e^{\nu+\lambda}>0$, the identity $e^{-2\lambda}=\gb\,e^{-(\nu+\lambda)}$ yields
\[
\lim_{u\uparrow u_*}e^{-2\lambda(u,r_*)}
=\lim_{u\uparrow u_*}\gb(u,r_*)\,e^{-(\nu+\lambda)(u,r_*)}=0,
\]
which contradicts the positive lower bound. This contradiction completes the proof.
\ese
\end{proof}

%-----------------------------------------------------------------------------------------------------------------------

\subsection{Bondi mass}

We assume the existence of the limit\footnote{  
In the present characteristic setting, finiteness of $M_0$ ensures that the Bondi mass is well-defined and that the integral identities used in Propositions~\ref{prop Bondi} and thereafter are meaningful. Without this assumption, the spacetime would not represent an isolated gravitating system but rather a configuration with infinite total energy, for which neither Bondi mass nor energy flux identities are physically appropriate.}
\be
M_0 := \lim_{r\to + \infty}m(0,r),
\ee
that is, the initial total mass is finite. Then we set 
\be
M(u) := \lim_{r\to + \infty}m(u,r), 
\ee
which is referred to as the \emph{Bondi mass}. 

\begin{proposition}[Generalized Bondi theorem for $f(R)$ gravity]
\label{prop Bondi}
The Bondi mass $M$ is a monotonically non-increasing function of the variable $u$: 
\be
\frac{d M}{d u}\leq 0.
\ee
\end{proposition}

\begin{proof}
\bse
Fix some 
$
0\leq u_1\leq u_2.
$
For any $r \geq 0$, assume that the characteristic $\chi_1(\cdot, r_0)$ passes the point $(u_2,r)$. From Proposition \ref{prop decreasing of m along D}, we have 
\be
m(\chi_1(u_1))\geq m(\chi_1(u_2))
\ee
and, from Proposition \ref{prop increasing m along r},
\be
M(u_1)\geq m(\chi_1(u_2)) = m(u_2,r),
\ee
which is valid for any $r$. Since
\be
M(u) := \lim_{r\to + \infty} m(r,u),
\ee
we conclude that
${M(u_2) \leq M(u_1)}$.
\ese
\end{proof}

%------------------------------------------------------------------------------------------------------------------

\subsection{Final mass}

Consider a sufficiently regular and globally defined solution 
to the system \eqref{notresysteme}. 
From Proposition \ref{prop Bondi}, the limit
\be
M_f = \lim_{u\to + \infty} M(u)
\ee
exists, which is referred to as the {\em final Bondi mass}.

\begin{proposition}[A property of the final Bondi mass]
For $r_0 > 2 M_f$, the timelike lines $r = r_0$ are complete towards the future.
\end{proposition}

\begin{proof}
\bse
Assume $0\leq u_1 \leq u_2.$
We write down the proper time between $(u_1,r_0)$ and $(u_2,r_0)$ along the timelike line $r=r_0$ as
\be
\tau_{(u_1,u_2)} = \int_{u_1}^{u_2}e^{\nu(u,r_0)}du.
\ee
We now claim that
\be
\lim_{u_2\to + \infty} \tau_{(0,u_2)} = + \infty.
\ee
To prove this, we proceed with the following estimates. First of all, from the proof of Proposition \ref{prop mass positive}, we get
$ 
e^{\nu- \lambda} \leq e^{\nu+\lambda}
$
which leads us to
$  
e^{\nu} \geq e^{\nu- \lambda}.
$
\ese

\bse
Now, we focus on $\nu- \lambda$:
\be
\aligned
(\nu- \lambda)|_{(u_1,r_0)}
& = - \int_{r_0}^{\infty}\del_r(\nu- \lambda)|_{(u_1,r)}dr
   =- \int_{r_0}^{\infty}\del_r(\ln(e^{\nu- \lambda}))|_{(u_1,r_0)} 
\\
& = - \int_{r_0}^{\infty}\frac{e^{\nu+\lambda}}{re^{\nu- \lambda}}\Big(1- \frac{e^{\nu- \lambda}}{e^{\nu+\lambda}}- r^2e^{-2\kappa\rho}(\Vstar(\rho)+8\pi U(\phi))\Big) \, dr
\\
& \geq - \int_{r_0}^{\infty}\frac{e^{\nu+\lambda}}{re^{\nu- \lambda}}\Big(1- \frac{e^{\nu- \lambda}}{e^{\nu+\lambda}}\Big) dr
 = - \int_{r_0}^{\infty} \frac{2m}{r^2}\Big(1- \frac{2m}{r}\Big)^{-1}dr.
\endaligned
\ee
Since $m(u,r)\leq M(u)$, we find 
\be
\aligned
(\nu- \lambda)|_{(u_1,r_0)} \geq - \int_{r_0}^{\infty} \frac{2M(u_1)}{r^2}\Big(1- \frac{2M(u_1)}{r}\Big)^{-1}dr
= \log\Big(1- \frac{2M(u_1)}{r_0}\Big),
\endaligned
\ee
that is, 
\be
e^{\nu(u_1,r_0)} \ge 1- \frac{2M(u_1)}{r_0}.
\ee
For $r_0\geq 2M_f$, as $M(u)\to M_f$ decreasingly, there exists a sufficiently large $u_2$  such that for any $u\geq u_2$
\be
1- \frac{2M(u)}{r_0}\geq \frac{1}{2}\Big(1- \frac{2M_f}{r_0}\Big).
\ee
For $u\geq u_2$, when $u\to + \infty$  we find 
\be
\tau_{(0,u)} = \int_0^{u}e^{\nu(\tau,r_0)}d\tau
\geq \int_{u_2}^{u}e^{\nu(\tau,r_0)}d\tau
\geq \frac{1}{2}\Big(1- \frac{2M_f}{r_0}\Big) (u-u_2) \to + \infty.
\qedhere
\ee
\ese
\end{proof} 
 
%================================================================================

\section{  Conclusions and perspectives}
\label{sec:conclusions}

\subsection{  Structural reduction and consistency}

{  

We introduced a characteristic first-order formulation of the $f(R)$ field equations coupled to a (possibly massive) scalar field in spherical symmetry, using generalized Bondi--Sachs coordinates. 

The main structural difficulty in $f(R)$ gravity lies in the presence of higher-derivative terms, in particular derivatives of the scalar curvature. In standard characteristic gauges, these terms prevent the system from closing at first order and obscure the hyperbolic structure. 

Our main contribution is the construction of an \emph{augmented characteristic gauge} based on an augmented conformal formulation. Following the treatment of the theory of $f(R)$ gravity in the Minkowski regime, proposed earlier by LeFloch and Ma \cite{LeFlochMa17a} --\cite{LeFlochMa26}, 
we introduced the conformal metric and the variable
\be
\rho := \kappa^{-1} \log f'(R),
\ee
which, under the standard viability condition $f'(R)>0$, can be treated as an \emph{independent unknown}. This augmentation is not merely a change of variables: it is precisely what allows the characteristic reduction to close at first order while remaining equivalent to the full $f(R)$ equations.

We proved that, within a mild $C^1$ regularity class and under appropriate center regularity conditions, solutions of the reduced integro-differential system for $(\phi,\rho)$ are equivalent to solutions of the full $f(R)$ system in spherical symmetry. In particular, the \emph{nonlinear differential constraint $f'(R)=e^{\kappa\rho}$ propagates}, so that the augmented formulation does not enlarge the theory but provides a mathematically consistent reformulation.

}

\subsection{  Geometric control and Hawking mass monotonicity}

{  
Beyond providing a closed evolution system, the characteristic reduction yields a robust geometric control mechanism. In the Einstein--scalar-field setting, monotonicity properties of the Hawking mass play a central role in collapse analysis and global arguments. 

In the $f(R)$ framework, an additional obstruction arises: there is no \emph{a priori} guarantee that the constraint reconstructing the metric coefficient $e^{\nu-\lambda}$ preserves its sign. The augmented formulation allows us to control this issue through a Hawking-mass identity, which ensures that the relevant metric factors remain globally well-defined \emph{under the structural assumptions} adopted in this work.

We established that the Hawking mass is non-decreasing along radial directions and non-increasing along incoming null directions. These monotonicity properties are geometric in nature and provide a priori control relevant to trapped-surface formation, collapse scenarios, and strong-field regimes. In numerical applications, they also furnish natural diagnostics for monitoring the integrity of the evolution and the onset of nonlinear effects.

}

\subsection{  Perspectives and open directions}

{  

At a formal level, the present formulation recovers Christodoulou's characteristic system in the \emph{formal} limit $f(R)\to R$ and $U(\phi)\to 0$, providing a consistency check on the reduction. Rigorous analysis of the near-Einstein regime without symmetry restriction can be found in \cite{LeFlochMa23,LeFlochMa26}.

The scope of the present manuscript is primarily mathematical. Our objective was, in spherical symmetry, to construct a characteristic first-order formulation and to establish its properties. A detailed confrontation with observational constraints would require selecting specific $f(R)$ models and parameter regimes, incorporating also a more realistic matter model.

The formulation developed here naturally suggests several directions for further investigations. 

\smallskip
\noindent\emph{(i) Global evolution and collapse.}
The reduced system provides a framework for studying global existence, trapped-surface formation, and strong-field regimes beyond the near-Minkowski setting. The interaction between the additional scalar degree of freedom and the matter field may lead to collapse mechanisms not present in the Einstein case. It would be natural to seek an extension of Christodoulou's pioneering work on Penrose's weak cosmic censorship and the formatipon of trapped surfaces for Einstein gravity. 

\smallskip
\noindent\emph{(ii) Characteristic numerical implementation.}
The structure ``evolve $(\phi,\rho)$ along null directions and reconstruct the metric by radial integration'' is directly compatible with characteristic evolution codes. Implementing this system for specific $f(R)$ models would allow quantitative exploration of nonlinear dynamics and comparison with other formulations.

\smallskip
\noindent\emph{(iii) Extensions to broader modified-gravity models.}
The augmented-variable strategy ---treating higher-derivative geometric quantities as independent unknowns to obtain a closed first-order characteristic system--- may extend to more general modified Lagrangians and matter models.

\smallskip
Overall, the augmented characteristic gauge introduced here provides a mathematically consistent and geometrically controlled first-order reduction of the $f(R)$--scalar-field system in spherical symmetry. It bridges structural PDE analysis and characteristic evolution techniques, and it lays groundwork for future studies of nonlinear dynamics and global geometry in modified gravity.

}

%=============================================================================

\paragraph*{Acknowledgments.}

This work was initiated while PLF was a visiting research fellow at the Courant Institute of Mathematical Sciences, New York University, and a visiting professor at the School of Mathematical Sciences, Fudan University, Shanghai. FCM is grateful to CAMGSD, IST-ID, for support through the FCT/Portugal project UIDB/04459/2020 as well as to CMAT, Univ. Minho, through the project UIDB/00013/2020 and FEDER funds (COMPETE). Both authors benefited from the support of the H2020-MSCA-2022-SE project Einstein--Waves, Grant Agreement~101131233. PLF was also supported by the research project \emph{Einstein-PPF}: ANR-23-CE40-0010, funded by the Agence Nationale de la Recherche (ANR).

%==========================================================================

%===============================================================================

\appendix

\section{Connection coefficients and curvature components}
\label{section---A1}

The connection associated with the metric \eqref{eq:metricrad} is torsion-free and we have
\be
(\nablat_\alpha e_{\beta},e_{\gamma}) = \frac{1}{2}\big(([e_\alpha,e_{\beta}],e_{\gamma})+([e_{\gamma},e_\alpha],e_{\beta})+([e_{\gamma},e_{\beta}],e_\alpha)\big).
\ee
The commutators of the frame \eqref{frame1} read
\be
\aligned
& [e_0,e_1] = [n,l] = e^{- \lambda- \nu}\big(\del_r\nu\del_u- \del_u\lambda\del_r\big),
\qquad
&& [e_0,e_2] = [n,\zeta_1] = \frac{1}{2re^{\lambda}}\zeta_1,
\\
& [e_0,e_3] = [n,\zeta_2] = \frac{1}{2re^{\lambda}}\zeta_2,
\qquad
&& [e_1,e_2] = [l,\zeta_1] = - \frac{1}{re^{\lambda}}\zeta_1,
\\
&[e_1,e_3] = [l,\zeta_2] = - \frac{1}{re^{\lambda}}\zeta_2,
\qquad
&&[e_2,e_3] = [\zeta_1,\zeta_2] = -r^{-1} (\cot\theta ) \zeta_2.
\endaligned
\ee
The non-vanishing components of the connection coefficients of the frame are
\be
\aligned
&(\nablat_n n,l) = -e^{- \nu}\del_u\lambda + \frac{1}{2}e^{- \lambda}\del_r \nu,
\qquad
&&(\nablat_n l,n) = e^{- \nu}\del_u\lambda - \frac{1}{2}e^{- \lambda}\del_r \nu,
\\
&(\nablat_l n,l) = e^{- \lambda}\del_r \nu,
\qquad
&&(\nablat_l l,n) = -e^{- \lambda}\del_r\nu,
\\
&(\nablat_{\zeta_1}n,\zeta_1) = - \frac{1}{2re^{\lambda}},
\qquad
&&(\nablat_{\zeta_1}l,\zeta_1) = \frac{1}{re^{\lambda}},
\\
&(\nablat_{\zeta_1}\zeta_1,n) = \frac{1}{2re^{\lambda}},
\qquad
&&(\nablat_{\zeta_1}\zeta_1,l) = - \frac{1}{re^{\lambda}},
\\ 
&(\nablat_{\zeta_2}n,\zeta_2) = - \frac{1}{2re^{\lambda}},
\qquad
&&(\nablat_{\zeta_2}l,\zeta_2) = \frac{1}{re^{\lambda}},
\\
&(\nablat_{\zeta_2}\zeta_1,\zeta_2) = r^{-1}\cot\theta,
\qquad
&&(\nablat_{\zeta_2}\zeta_2,n) = \frac{1}{2 re^{\lambda}},
\\
&(\nablat_{\zeta_2}\zeta_2,l) = - \frac{1}{re^{\lambda}},
\qquad
&&(\nablat_{\zeta_2}\zeta_2,\zeta_1) = -r^{-1}\cot\theta,
\endaligned
\ee
while the covariant derivatives of the frame vectors are
\be
\aligned
& \nablat_n n = \Big(e^{- \nu}\del_u\lambda - \frac{1}{2}e^{- \lambda}\del_r\nu\Big)l,
\qquad
&& \nablat_n l = \Big(\frac{1}{2}e^{- \lambda}\del_r\nu - e^{- \nu}\del_u\lambda\Big)l,
\\
& \nablat_l n = -e^{- \lambda}\del_r\nu \cdot n,
\qquad
&& \nablat_l l = e^{- \lambda}\del_r\nu\cdot l,
\\
& \nablat_{\zeta_1}n = - \frac{1}{2re^{\lambda}}\cdot \zeta_1,
\qquad
&& \nablat_{\zeta_1}l = \frac{1}{re^{\lambda}}\zeta_1,
\\
& \nablat_{\zeta_1}\zeta_1 = \frac{1}{re^{\lambda}}n - \frac{1}{2re^{\lambda}}l, 
\qquad
&& \nablat_{\zeta_1}\zeta_2 = -r^{-1}\cot\theta\cdot\zeta_1,
\\
& \nablat_{\zeta_2}n = - \frac{1}{2re^{\lambda}}\zeta_2,
\qquad
&& \nablat_{\zeta_2}l = \frac{1}{re^{\lambda}}\zeta_2,
\\
& \nablat_{\zeta_2}\zeta_1 = r^{-1}\cot \theta\cdot \zeta_2,
\qquad
&& \nablat_{\zeta_2}\zeta_2 = \frac{1}{re^{\lambda}}n - \frac{1}{2re^{\lambda}}l - r^{-1}\cot\theta\cdot\zeta_1.
\endaligned
\ee
We can now compute the covariant derivatives of the coframe \eqref{coframe} by
$\del_\alpha\big(<\omega^\beta,e_{\gamma}>\big) = 0$, thus
$<\nablat_\alpha\omega^\beta,e_{\gamma}> + <\omega^\beta,\nablat_\alpha e_{\gamma}> = 0$, which shows that
\bel{eq:520}
\nablat_\alpha\omega^\beta = -<\omega^\beta,\nablat_\alpha e_{\gamma}>\omega^{\gamma}.
\ee
For any  spherically symmetric function $w$, we can write $dw = n(w)\omega^0 + l(w)\omega^1$, so that the Hessian reads
\be
\aligned
\nablat dw
& =  <n,n(w)>\omega^0\otimes\omega^0 + <n,l(w)>\omega^1\otimes\omega^0 + <l,n(w)>\omega^0\otimes\omega^1
\\
& \quad +  <l,l(w)>\omega^1\otimes\omega^1 + n(w)\nablat\omega^0 + l(w)\nablat\omega^1.
\endaligned
\ee
We thus obtain the following expression in terms of the metric coefficients in Bondi coordinates:
\be
\aligned
\nablat_n\nablat_n w
& = e^{-2\nu}\del_{uu}w + \frac{1}{4}e^{-2\lambda}\del_{rr}w - e^{- \nu- \lambda}\del_{ur}w
+
\big(\frac{1}{2}e^{- \nu- \lambda}\del_r\nu - e^{-2\nu}\del_u\nu\big)\del_u w
\\
& \quad +\Big(\frac{1}{2}e^{-2\lambda}\del_r\nu - \frac{1}{2}e^{- \nu- \lambda}\del_u\lambda- \frac{1}{4}e^{-2\lambda}\del_r\lambda\Big)\del_r w,
\\
\nablat_n\nablat_l w
& = \nablat_l\nablat_n w = - \frac{1}{2}e^{-2\lambda}\del_r(\nu- \lambda)\del_rw - \frac{1}{2}e^{-2\lambda}\del_{rr}w + e^{- \nu- \lambda}\del_{ru}w,
\\
\nablat_l\nablat_l w
& = e^{-2\lambda}(\del_{rr}w- \del_r\lambda\del_rw- \del_r\nu\del_rw),
\\
\nablat_{\zeta_1} \nablat_{\zeta_1}w
& = \nablat_{\zeta_2} \nablat_{\zeta_2}w = -r^{-1}e^{- \nu- \lambda}\del_uw + r^{-1}e^{-2\lambda}\del_rw.
\endaligned
\ee
By taking the trace of the Hessian, we immediately obtain the expression of the wave operator acting on a spherically symmetric function. 
In turn, the Christoffel  symbols in the coordinates $(u,r,\theta,\varphi)$ read:
\be
\aligned
& \Gamma^u_{uu}= \del_u(\nu+\lambda)-e^{\nu- \lambda}\del_r\nu,\qquad
&& \Gamma^{r}_{uu}=-e^{\nu- \lambda}\del_u\lambda+e^{2(\nu- \lambda)}\del_r\nu,
\\
& \Gamma^{r}_{ur}= \Gamma^{r}_{ru}=e^{\nu- \lambda}\del_r\nu,\qquad
&& \Gamma^{r}_{rr}= \del_r(\nu+\lambda), 
\\
& \Gamma^{\theta}_{r\theta}= \Gamma^{\theta}_{\theta r}= \frac{1}{r},\qquad
&& \Gamma^{\theta}_{\varphi\varphi}=- \frac{1}{2}\sin{2\theta}, 
\\
& \Gamma^{u}_{\theta\theta}=re^{- \nu- \lambda}, \qquad
&& \Gamma^{r}_{\theta\theta}=-re^{-2\lambda}, 
\\
& \Gamma^{\varphi}_{r\varphi}= \Gamma^{\varphi}_{\varphi r}= \frac{1}{r},\qquad
&& \Gamma^{\varphi}_{\varphi\theta}= \Gamma^{\varphi}_{\theta\varphi}= \cot{\theta}, 
\\
& \Gamma^{u}_{\varphi\varphi}=r\sin^{2}{\theta}e^{- \nu- \lambda},
\qquad
&& \Gamma^{r}_{\varphi\varphi}=-r\sin^{2}{\theta}e^{-2\lambda}. 
\endaligned
\ee
In the frame $(n,l,\zeta_1,\zeta_2)$ we find
\be
\aligned
\label{connection-on-frame}
& \Gamma^{n}_{nn}=e^{- \nu}\del_u\lambda- \frac{1}{2}e^{- \lambda}\del_r\nu,
\qquad
&& \Gamma^{n}_{\zeta_1\zeta_1}= \frac{e^{- \lambda}}{r}, 
&&& \Gamma^{n}_{\zeta_2\zeta_2}= \frac{e^{- \lambda}}{r}, 
\\
& \Gamma^{n}_{ln}=-e^{- \lambda}\del_r\nu,
&& \Gamma^{l}_{nl}= \frac{1}{2}e^{- \lambda}\del_r\nu-e^{- \nu}\del_u\lambda,\qquad
&&& \Gamma^{l}_{ll}=e^{- \lambda}\del_r\nu,
\\
& \Gamma^{l}_{\zeta_1\zeta_1}=- \frac{e^{- \lambda}}{2r},\qquad
&& \Gamma^{l}_{\zeta_2\zeta_2}=- \frac{e^{- \lambda}}{2r},
&&& \Gamma^{\zeta_1}_{\zeta_1n}=- \frac{e^{- \lambda}}{2r}, 
\\
& \Gamma^{\zeta_1}_{\zeta_1l}= \frac{e^{- \lambda}}{r},
&& \Gamma^{\zeta_1}_{\zeta_2\zeta_2}=- \frac{1}{r}\cot{\theta}, \qquad
&&& \Gamma^{\zeta_2}_{\zeta_2n}=- \frac{e^{- \lambda}}{2r}, 
\\
& \Gamma^{\zeta_2}_{\zeta_2l}= \frac{e^{- \lambda}}{r}, \qquad
&& \Gamma^{\zeta_2}_{\zeta_2\zeta_1}= \frac{\cot{\theta}}{r}.
&&&
\endaligned
\ee
Then the non-vanishing frame components of the Ricci curvature of \eqref{eq:metricrad} read
\bel{eq:JDJK9}
\aligned
\Rt_{nn} & = \frac{1}{r}\Big(\frac{1}{2}e^{-2\lambda}\del_r(\nu+\lambda) - 2e^{- \nu- \lambda}\del_u\lambda\Big),
\\
\Rt_{nl} & = \Rt_{ln} = -e^{- \nu- \lambda}\del_{ur}(\nu+\lambda) + \frac{e^{-2\lambda}}{r}\del_r(\nu- \lambda) + e^{-2\lambda}\big(\del_{rr}\nu + (\del_r\nu)^2 - \del_r\lambda\del_r\nu\big),
\\
\Rt_{ll} & = \frac{2e^{-2\lambda}}{r}\del_r(\nu+\lambda),
\qquad\qquad
\Rt_{\zeta_1\zeta_1}
 = \Rt_{\zeta_2\zeta_2} = \frac{1}{r^2}(1-e^{-2\lambda}) - \frac{e^{-2\lambda}}{r}\del_r(\nu- \lambda).
\endaligned
\ee 
Taking the trace of the conformal Ricci tensor with respect to $\gt$, we find the expression $
\Rt=-2\Rt_{nl} + \Rt_{\zeta_1\zeta_1} + \Rt_{\zeta_2\zeta_2}
$
for the scalar curvature and, therefore, 
\bel{eq:Rtilde}
\Rt=
 \frac{2}{e^{\nu+\lambda}}\del_{ur}(\nu+\lambda) - \frac{4}{re^{2\lambda}}\del_r(\nu- \lambda)
 - \frac{2}{e^{2\lambda}}(\del_{rr}\nu + \del_r\nu \del_r (\nu- \lambda)) + \frac{2}{r^2}\big(1-e^{-2\lambda}\big).
\ee
We then deduce the expressions of the frame components of the Einstein tensor of the conformal metric, i.e. 
\be
\label{G-components}
\aligned
\Gt_{nn}
& = \Rt_{nn},
\qquad
\Gt_{ll}
 = \Rt_{ll},
\qquad
\Gt_{nl}
= \Rt_{\zeta_1\zeta_1}= \Rt_{\zeta_2\zeta_2},
\qquad 
\Gt_{\zeta_1\zeta_1}
= \Gt_{\zeta_2\zeta_2} = \Rt_{nl}, 
\endaligned
\ee
which are given by \eqref{eq:JDJK9}.  

%===========================

%-----------------------------------------------------------------------------------------------------------

\section{ Proof of two technical lemmas} 
\label{append-twolemmas}

\begin{proof}[Proof of Lemma~\ref{sec 6 lem R_g apriori}] 
\bse
Recall that the scalar curvature of the metric $\gt$ is denoted by $\widetilde R$. From \eqref{R-conformal-transf}, the scalar curvature of $g$ can be expressed in terms of $\widetilde R$
and the terms $e^{\kappa\rhoh},\Box_{\gt}\rhoh$ and $\gt(\nablat\rhoh,\nablat\rhoh)$, which are well defined and bounded on $[0,u_0]\times[0,+ \infty)$.  So, our strategy is to check that $\widetilde R$ is well defined when $r\neq 0$. The expression of $\widetilde R$ in terms of $\nu$ and $\lambda$ is given by \eqref{eq:Rtilde}, which can be rewritten as
\bel{sec 6 eq_expression Rt_g}
e^{\nu+\lambda}\widetilde R = 2\del_{ur}(\nu+\lambda) - \frac{4}{r}\del_re^{\nu- \lambda} - 2e^{\nu- \lambda}\del_{rr}\nu - 2e^{\nu- \lambda}\del_r\nu\del_r(\nu- \lambda) + \frac{2(e^{\nu+\lambda} - e^{\nu- \lambda})}{r^2}.
\ee
Hence, we need to check that $\del_{rr}(\nu+\lambda),\,\del_{ur}(\nu+\lambda)$ and $r\del_{rr}(\nu- \lambda)$ can be continuously extended on $[0,u_0]\times [0,\infty)$.

We begin with $\del_{rr}(\nu+\lambda)$: Recalling that $\del_r(r\phi)$ is of class $C^1$ on $[0,u_0]\times[0,\infty)$, we have
\be
r\del_{rr}\phi = \del_{rr}(r\phi) - 2\del_r\phi.
\ee
So, $r\del_{rr}\phi \in C\big([0,u_0]\times[0,\infty)\big)$ and, using similar arguments, we get that $r\del_{rr}\rhoh \in C\big([0,u_0]\times[0,\infty)\big)$.
%.  
Since the term on the right-hand side of the first equation in \eqref{Bondi system diff augmented1.5} is differentiable with respect to $r$ on $[0,u_0]\times(0,\infty)$, then $\del_r(\nu+\lambda)$ is differentiable with respect to $r$ on $[0,u_0]\times(0,\infty)$ and  
\be
\del_{rr}(\nu+\lambda) = \frac{3}{4}\kappa^{2}\Big(|\del_r\rhoh|^2 + 2r\del_r\rhoh\del_{rr}\rhoh\Big) + \frac{4\pi}{e^{\kappa\rhoh}}\Big(|\del_r\phi|^2- \kappa r\del_r\rhoh|\del_r\phi|^2 + 2r\del_r\phi\del_{rr}\phi\Big),
\ee
so the aforementioned term can be continuously extended on $[0,u_0]\times[0,\infty)$, as claimed.
\ese

\bse
Regarding the term $r\del_{rr}(\nu- \lambda)$, observe that from the second equation of \eqref{Bondi system diff augmented1.5}, we obtain
\be
r\del_{rr}(\nu- \lambda) = \del_r\Big[\Big(1-r^2e^{-2\kappa\rhoh}\big(\Vstar(\rhoh)+8\pi U(\phi)\big)\Big)e^{2\lambda}\Big]- \del_r(\nu- \lambda) \in C\big([0,u_0]\times[0,\infty)\big).
\ee
For the term $\del_{ur}(\nu+\lambda)$, we need the following identities involving the mixed derivatives of $\rhoh$ and $\phi$:
\be
\aligned
r\del_{ur}\phi = D(\del_r(r\phi)) + \frac{1}{2}e^{\nu- \lambda}\del_{rr}(r\phi) - \del_u\phi \in C\big([0,u_0]\times[0,\infty)\big),
\\
r\del_{ur}\rhoh = D(\del_r(r\rhoh)) + \frac{1}{2}e^{\nu- \lambda}\del_{rr}(r\rhoh) - \del_u\rhoh \in C\big([0,u_0]\times[0,\infty)\big).
\endaligned
\ee
{  These identities follow directly from the definition $D=\partial_u-\frac12 e^{\nu-\lambda}\partial_r$ and the commutation of mixed partial derivatives: for instance, $\partial_u\partial_r(r\phi)=D(\partial_r(r\phi))+\frac12 e^{\nu-\lambda}\partial_{rr}(r\phi)$.}
Then, we have 
\be
\del_{ur}(\nu+\lambda) = \frac{3}{2}\kappa^{2}r\del_r\rhoh\del_{ur}\rhoh - \frac{4\pi r}{e^{\kappa\rhoh}}\Big(\kappa\del_u\rhoh|\del_r\phi|^2 + 2\del_r\phi\del_{ur}\phi\Big) \in C\big([0,u_0]\times[0,\infty)\big),
\ee
which completes the proof. 
\ese
\end{proof}

%----------------------------------------------------------------------

\begin{proof}[Proof of Lemma \ref{sec 6 lem R_g main estimate}]
\bse
 \noindent\emph{Step 1 (Divergence of the stress-energy tensor).}
As before, the stress-energy  tensor  reads
$T_{\alpha\beta} = \del_{\alpha}\phi\del_{\beta}\phi - \left(\frac{1}{2}\sigma+U(\phi)\right)g_{\alpha\beta}$, but now with 
\be
\sigma=e^{\kappa\rhoh}\left(-2e^{\nu- \lambda}\del_{u}\phi\del_{r}\phi+e^{-2\lambda}(\del_{r}\phi)^{2}\right),
\ee
 so that its
non-vanishing components (already computed in \eqref{eq Bondi components T}) are 
\bel{eq Bondi components T with hat}
\aligned
& T_{nn} = \big(e^{- \nu}\del_u \phi - \frac{1}{2}e^{- \lambda}\del_r\phi\big)^2,
\qquad \qquad
T_{ll} = e^{-2\lambda}(\del_r\phi)^2,
\\
& T_{\zeta_1\zeta_1} = T_{\zeta_2\zeta_2} = - e^{- \kappa\rhoh}\left(\frac{\sigma}{2}+U(\phi)\right),
\quad~ T_{nl}=T_{ln}=e^{- \kappa\rhoh}U(\phi),
\endaligned
\ee
with the augmented variable $\rhoh$ now appearing in some of the components.

Observe that, by the conformal transformation  
\be
\label{block-g}
\square_g \phi=e^{\kappa \rhoh}  \square_{\widetilde g}\phi- \kappa g(\nabla \rhoh,\nabla \phi)
= e^{\kappa \rhoh} \Big( \square_{\widetilde g}\phi+\kappa e^{- \nu- \lambda}(\del_r\phi\, D\rhoh+\del_r\rhoh\, D\phi)\Big)
\ee
together with \eqref{Bondi d'Alembert} and \eqref{eq:metricrad}, we obtain
\be
\aligned
\Box_g \phi & =   
- \frac{2}{r}e^{\kappa\rhoh}e^{- \nu- \lambda}\Big(D(\del_r(r\phi)) -  {r \over 2} \del_r(e^{\nu- \lambda})\del_r\phi - \frac{1}{2}\kappa r(D\rhoh\del_r\phi + D\phi\del_r\rhoh)\Big).
\endaligned
\ee
Then, recalling that $\nabla^{\alpha}T_{\alpha\beta} = \del_{\beta}\phi \left(\Box_g\phi-U'(\phi)\right)$ and comparing the expressions above with the third equation of the augmented system \eqref{Bondi system diff augmented1.5}, we conclude that
$\Box_g \phi = U'(\phi)$, so that $\nabla^{\alpha}T_{\alpha\beta} = 0$.
\ese
%

%--------------------------------------------------------------------------------------------

\bse
\noindent\emph{Step 2 (Components $(n,l)$ and $(l,l)$ of the augmented system).}
In view of \eqref{eq Bondi components T with hat}, the first two equations in \eqref{Bondi system diff augmented1.5} and the components \eqref{G-components}, we find
\bel{Bondi augmentation1 proof 1}
\aligned
& \Gt_{ll} = \frac{3}{2}\kappa^{2}e^{-2\lambda}|\del_r\rhoh|^2 + 8\pi e^{- \kappa\rhoh}T_{ll} ,
\qquad
\quad 
\Gt_{nl} = e^{-2\kappa\rhoh}\big(\Vstar(\rhoh)+8\pi U(\phi)\big).
\endaligned
\ee
Now, let us introduce the following tensor of ``augmented modified gravity''
\bel{Bondi augmentation1 proof tenser E}
\aligned
E^{\rhoh}_{\alpha\beta}
:& =  e^{\kappa\rhoh}R_{\alpha\beta} - \frac{1}{2}f(\Rstar(\rhoh))g_{\alpha\beta} +\big(g_{\alpha\beta}\Box_g - \nabla_{\alpha}\nabla_{\beta}\big)e^{\kappa\rhoh}
\\
& =  e^{\kappa\rhoh}G_{\alpha\beta} + e^{- \kappa\rhoh}\Vstar(\rhoh)\gt_{\alpha\beta}
+ \big(g_{\alpha\beta}\Box_g - \nabla_{\alpha}\nabla_{\beta}\big)e^{\kappa\rhoh} .
\endaligned
\ee\label{Bondi augmentation1 proof 1.5}
By a conformal transformation, the tensor $E^{\rhoh}_{\alpha\beta}$ can also be expressed in terms of $\gt$, namely 
\bel{Bondi augmentation1 proof 2} 
E^{\rhoh}_{\alpha\beta} = e^{\kappa\rhoh}\Gt_{\alpha\beta} - \frac{3}{2}\kappa^{2}e^{\kappa\rhoh}\del_{\alpha}\rhoh\del_{\beta}\rhoh + \frac{3}{4}\kappa^{2}\gt_{\alpha\beta}e^{\kappa\rhoh}\gt(\nablat\rhoh,\nablat\rhoh) + \gt_{\alpha\beta}e^{- \kappa\rhoh}\Vstar(\rhoh).
\ee
We include also the matter content and, as before, in agreement with the notation of \eqref{eq:939k}, we introduce the tensor
\bel{tensor F}
F^{\rhoh}_{\alpha\beta} := E^{\rhoh}_{\alpha\beta} - 8\pi T_{\alpha\beta}.
\ee
Now, by \eqref{eq Bondi components T with hat}, \eqref{G-components} and  \eqref{Bondi augmentation1 proof 2}, the non-vanishing frame components of $F^{\rhoh}$ are 
$F^{\rhoh}_{nn}$,
$F^{\rhoh}_{nl}$, 
$F^{\rhoh}_{ll},$
and 
$F^{\rhoh}_{\zeta_1\zeta_1} = F^{\rhoh}_{\zeta_2\zeta_2}$.
Combining \eqref{Bondi augmentation1 proof 1} with \eqref{Bondi augmentation1 proof 2}, we obtain easily that
\bel{Bondi augmentation1 proof 3}
F^{\rhoh}_{nl}\equiv 0,
\qquad
F^{\rhoh}_{ll} \equiv 0,
\ee
so these components provide no new information.
\ese
%

%------------------------------------------

\vskip.3cm

\bse
\noindent\emph{Step 3 (A useful identity).} We claim that
the tensor $E_{\alpha\beta}^{\rhoh}\omega^{\alpha}\otimes\omega^{\beta}$, defined in \eqref{Bondi augmentation1 proof tenser E}, satisfies
\bel{Bondi augmentation1 proof 4}
\nabla^{\alpha}E^{\rhoh}_{\alpha\beta} = \frac{1}{2}e^{\kappa\rhoh}\del_{\beta}(R - \Rstar(\rhoh)),
\ee
which is an essential identity for our purposes. This can be checked by a direct calculation based on the identities derived in the proof of Lemma 2.3.1 of LeFloch and Ma~\cite{LeFlochMa17a}, which leads us to
\be
\aligned
\nabla^{\alpha}\big(e^{\kappa\rhoh} R_{\alpha\beta}\big) & =e^{\kappa\rhoh}\nabla^{\alpha}R_{\alpha\beta} + \nabla^{\alpha}\big(e^{\kappa\rhoh}\big) R_{\alpha\beta}
 = \frac{1}{2}e^{\kappa\rhoh}\nabla_{\beta}R + \nabla^{\alpha}\big(e^{\kappa\rhoh}\big)R_{\alpha\beta}
\endaligned
\ee
and
\be
\aligned
\nabla^{\alpha}\big(f(\Rstar(\rhoh))g_{\alpha\beta}\big) & = \nabla_{\beta}\big(f(\Rstar(\rhoh))\big)
  = f'(\Rstar(\rhoh))\nabla_{\beta}\big(\Rstar(\rhoh)\big)
 =e^{\kappa\rhoh}\nabla_{\beta}\Rstar(\rhoh),
\endaligned
\ee
while  
\be
\nabla^{\alpha}\big((g_{\alpha\beta}\Box_g - \nabla_{\alpha}\nabla_{\beta})e^{\kappa\rhoh}\big)
=  \big(\Box_g\nabla_{\beta}- \nabla_{\beta}\Box_g\big)e^{\kappa\rhoh}
  -R_{\alpha\beta}\nabla^{\alpha}\big(e^{\kappa\rhoh}\big).
\ee
In view of \eqref{Bondi augmentation1 proof tenser E},
we can thus combine the identities above and get the desired result \eqref{Bondi augmentation1 proof 4}.
\ese
%

%------------------------------------------

\vskip.15cm

\bse
\noindent\emph{Step 4 (Components of the augmented system).}
We will now make use of the identity  \eqref{Bondi augmentation1 proof 4} in order to determine $F^{\rhoh}_{\zeta_1\zeta_1}$ and $F^{\rhoh}_{\zeta_2\zeta_2}$. Since we checked 
$\nabla^{\alpha}T_{\alpha\beta} = 0$,  we also have 
\bel{identity}
\nabla^{\alpha}F^{\rhoh}_{\alpha\beta} = \nabla^{\alpha}E^{\rhoh}_{\alpha\beta} = \frac{1}{2}e^{\kappa\rhoh}\del_{\beta}(R - \Rstar(\rhoh)).
\ee
Similarly to the proof of Proposition \ref{prop essential f(R) Bondi}, we use 
\be
\nabla_{\alpha}F^{\rhoh}_{\beta\gamma} = \langle \alpha, F^{\rhoh}_{\beta\gamma}\rangle - \Gamma_{\alpha\beta}^{\delta}F^{\rhoh}_{\delta\gamma} - \Gamma_{\alpha\gamma}^{\delta}F^{\rhoh}_{\beta\delta},
\ee
so that, after taking \eqref{Bondi augmentation1 proof 3} into account, we get
\be
\nabla^{\alpha}F^{\rhoh}_{\alpha l} = - \Gamma_{\zeta_1 l}^{\zeta_1}F^{\rhoh}_{\zeta_1\zeta_1}
- \Gamma_{\zeta_2 l}^{\zeta_2}F^{\rhoh}_{\zeta_2\zeta_2}
\ee
or, equivalently, 
\bel{Bondi augmentation1 proof 5}
 \Big(\frac{\kappa}{2}\del_r\rhoh - \frac{1}{r}\Big) F^{\rhoh}_{\zeta_1\zeta_1}= \frac{1}{4}e^{\kappa\rhoh}\del_r(R - \Rstar(\rhoh)).
\ee
 \ese

\vskip.3cm

%-------------------------------------------------------------------

\bse
\noindent\emph{Step 5 (Conclusion based on the trace).} We now compute the trace $g^{\alpha\beta}F^{\rhoh}_{\alpha\beta}$ with the help of 
\eqref{Bondi augmentation1 proof tenser E}, namely
\be
\aligned
g^{\alpha\beta} F^{\rhoh}_{\alpha\beta} & = g^{\alpha\beta}E^{\rhoh}_{\alpha\beta} - 8\pi g^{\alpha\beta}T_{\alpha\beta}
\\
& =g^{\alpha\beta}\big(e^{\kappa\rhoh}R_{\alpha\beta} - \frac{1}{2}g_{\alpha\beta}f(R_\star(\rhoh)) + (g_{\alpha\beta}\Box_g - \nabla_{\alpha}\nabla_{\beta})e^{\kappa\rhoh}\big) + 8\pi \sigma+32\pi U(\phi).
\endaligned
\ee
Then, using the conformal transformations \eqref{block-g} and 
\be
\label{box-exp}
\Box_{\gt} e^{\kappa\rho} = \kappa e^{\kappa\rho}\Big(\Box_{\gt}\rho+\kappa\gt (\nablat\rho,\nablat\rho)\Big)
\ee
we obtain
\be
\aligned
g^{\alpha\beta} F^{\rhoh}_{\alpha\beta}  
& = e^{\kappa\rhoh}R - 2f(\Rstar(\rhoh)) + 3\kappa e^{2\kappa\rhoh}\Box_{\gt}\rhoh + 8\pi \sigma + 32\pi U(\phi).
\endaligned
\ee
Recall that, by \eqref{box-tilde-rho},
\be
3\kappa e^{2\kappa\rhoh}\Box_{\gt}\rhoh - 2f(\Rstar(\rhoh)) + \Rstar(\rhoh)e^{\kappa\rhoh} + 8\pi \sigma+32\pi U(\phi) = 0,
\ee
then
\bel{Bondi augmentation1 proof 6}
g^{\alpha\beta} F^{\rhoh}_{\alpha\beta} = e^{\kappa\rhoh}(R - \Rstar(\rhoh)).
\ee
\ese

\bse
On the other hand, we have
\be
e^{- \kappa\rhoh}g^{\alpha\beta}F^{\rhoh}_{\alpha\beta} = \gt^{\alpha\beta}F^{\rhoh}_{\alpha\beta} = -F^{\rhoh}_{nl} - F^{\rhoh}_{ln} + F^{\rhoh}_{\zeta_1\zeta_1} + F^{\rhoh}_{\zeta_2\zeta_2}.
\ee
So, substituting \eqref{Bondi augmentation1 proof 3} and \eqref{Bondi augmentation1 proof 6} into this identity yields
\be
2F^{\rhoh}_{\zeta_1\zeta_1} = R - \Rstar(\rhoh), 
\ee
which, by \eqref{Bondi augmentation1 proof 5}, leads to 
\be
{ {\del_r(R- \Rstar(\rhoh)) = \Big(\kappa\del_r\rhoh - \frac{2}{r}\Big)e^{- \kappa\rhoh}\, (R- \Rstar(\rhoh)).}}
\ee
Finally, integrating in the radial direction gives
\be
\del_r\Big(\ln(R-\Rstar(\rhoh)) + e^{-\kappa\rhoh}\Big)
= -\frac{2}{r}e^{-\kappa\rhoh},
\ee
and therefore, for any $r_0>0$,
\be
\aligned
& (R-\Rstar(\rhoh))(u,r)\exp\!\big(e^{-\kappa\rhoh(u,r)}\big)
\\
& =
(R-\Rstar(\rhoh))(u,r_0)\exp\!\big(e^{-\kappa\rhoh(u,r_0)}\big)
\exp\!\Big(-\int_{r_0}^r\frac{2}{s}e^{-\kappa\rhoh(u,s)}\,ds\Big), 
\endaligned
\ee
equivalently
\be
S(u,r)=S(u,r_0)\exp\!\Big(-\int_{r_0}^r\frac{2}{s}e^{-\kappa\rhoh(u,s)}\,ds\Big),
\ee
with $S(u,r):=(R-\Rstar(\rhoh))(u,r)\exp(e^{-\kappa\rhoh(u,r)})$. This completes the proof of Lemma \ref{sec 6 lem R_g main estimate}.
\ese
\end{proof}

\end{document}